\documentclass[pra,twocolumn,a4paper,showpacs,aps,10pt,footinbib,allowtoday]{revtex4-2}

\usepackage{fontenc}
\usepackage[latin9]{inputenc}
\usepackage{amsmath}
\usepackage{amsthm}
\usepackage{amssymb}
\usepackage{stmaryrd}
\usepackage{graphicx}
\usepackage{fullpage}
\usepackage{xspace}
\usepackage{multirow}
\usepackage{rotating}
\usepackage[table]{xcolor}
\usepackage{mathrsfs}
\usepackage{mathtools}
\usepackage{braket}
\usepackage{upgreek}
\usepackage{comment}
\usepackage[normalem]{ulem}
\usepackage{caption}
\usepackage[hidelinks]{hyperref}
\usepackage{cleveref}
\usepackage{footmisc}
\usepackage{enumitem}

\makeatletter
\theoremstyle{definition}

\theoremstyle{plain}
\newtheorem{lem}{\protect\lemmaname}
\theoremstyle{plain}

\theoremstyle{plain}

\theoremstyle{plain}
\newtheorem{cor}{\protect\corollaryname}
\theoremstyle{plain}
\newtheorem{rem}{\protect\remarkname}

\DeclareMathOperator{\Tr}{Tr}

\@ifundefined{showcaptionsetup}{}{%
 \PassOptionsToPackage{caption=false}{subfig}}
\usepackage{subfig}
\makeatother

\providecommand{\corollaryname}{Corollary}
\providecommand{\definitionname}{Definition}
\providecommand{\lemmaname}{Lemma}
\providecommand{\propositionname}{Proposition}
\providecommand{\theoremname}{Theorem}
\providecommand{\remarkname}{Remark}

\definecolor{darkgreen}{rgb}{0,0.5,0}
\newcommand{\vicky}[1]{{\color{black} #1}}

\newcommand{\N}{\mathbb{N}}
\newcommand{\I}{\mathbb{I}}
\newcommand{\R}{\mathbb{R}}

\newcommand{\cH}{\mathcal{H}}
\newcommand{\cL}{\mathcal{L}}
\newcommand{\cQ}{\mathcal{Q}}
\newcommand{\cF}{\mathcal{F}}
\newcommand{\cA}{\mathcal{A}}

\newcommand{\cC}{\mathcal{C}}
\newcommand{\cX}{{X}}
\newcommand{\cY}{{Y}}
\newcommand{\bA}{\boldsymbol{A}}
\newcommand{\bB}{\boldsymbol{B}}
\newcommand{\tA}{\text{A}}
\newcommand{\tB}{\text{B}}
\newcommand{\cNS}{\mathcal{N\!S}}
\newcommand{\cNC}{\mathcal{N\!C}}

\newcommand{\pax}{p_{\tA}}
\newcommand{\paxax}{p_{\tA}(a|x)}

\newcommand{\tpa}{\hat{p}_\tA}
\newcommand{\tp}{p'}

\newcommand{\prlparagraph}[1]{\emph{#1.---}}

\DeclarePairedDelimiter\abs{\lvert}{\rvert}%
\DeclarePairedDelimiter\norm{\lVert}{\rVert}%

\makeatletter
\let\oldabs\abs
\def\abs{\@ifstar{\oldabs}{\oldabs*}}

\let\oldnorm\norm
\def\norm{\@ifstar{\oldnorm}{\oldnorm*}}
\makeatother

\crefformat{footnote}{#2\footenotemark[#1]#3}

\begin{document}
\title{An invertible map between Bell non-local and contextuality scenarios}

\author{Victoria J Wright}
\email{victoria.wright@icfo.eu}
\affiliation{ICFO-Institut de Ciencies Fotoniques, The Barcelona Institute of Science and Technology, 08860 Castelldefels, Spain}

\author{M\'at\'e Farkas}
\email{mate.farkas@york.ac.uk}
\affiliation{ICFO-Institut de Ciencies Fotoniques, The Barcelona Institute of Science and Technology, 08860 Castelldefels, Spain}
\affiliation{Department of Mathematics, University of York, Heslington, York, YO10 5DD}

\begin{abstract}
We present an invertible map between correlations in any bipartite Bell scenario and behaviours in a family of contextuality scenarios. The map takes local, quantum and non-signalling correlations to non-contextual, quantum and contextual behaviours, respectively. Consequently, we find that the membership problem of the set of quantum contextual behaviours is undecidable, the set cannot be fully realised via finite dimensional quantum systems and is not closed. Finally, we show that neither this set nor its closure is the limit of a sequence of computable supersets, due to the result MIP*=RE.
\end{abstract}

\maketitle

\prlparagraph{Introduction}
Bell non-locality \cite{nonlocality} describes correlations between space-like separated experiments that are impossible in any locally realistic theory. Such correlations are, however, allowed in quantum theory. Beyond their fundamental relevance these correlations have technological applications such as secure random number generation \cite{RNG} and cryptography \cite{DIQKD}. 

Generalised contextuality \cite{spekkens2005contextuality} similarly describes correlations that are absent from classical physics but instead of space-like separation, these correlations occur in experiments where there are \emph{operationally equivalent} experimental procedures. For example, two preparation procedures of a system are operationally equivalent if \vicky{every measurement on the system leads to the same statistics for both preparation procedures}. Contextual correlations have also found practical relevance, for example, in state discrimination~\cite{schmid2018contextual} and demonstrating quantum advantage in communication tasks~\cite{tavakoli2020measurement}.

One way to enforce an operational equivalence between preparations is by using the setup of a Bell non-locality experiment (known as a Bell scenario)\vicky{, under the assumption that no signal can travel faster than light}. \vicky{In a two-party Bell scenario two parties,} Alice and Bob\vicky{,} share a physical system. \vicky{In some frame of reference,} Alice selects and performs a measurement $x$ \vicky{from some pre-agreed options} on her \vicky{subsystem}, then \vicky{Bob measures his subsystem at a time before any light signal could have arrived. 

Under the no-signalling assumption,} the statistics Bob can observe  \vicky{from such a measurement} must not depend on $x$, otherwise \vicky{by performing this procedure with many shared systems simultaneously Bob} could infer Alice's choice $x$ \vicky{and a faster-than-light signal could be transmitted from Alice to Bob}. \vicky{It follows that viewing Alice's measurement of her subsystem as a preparation procedure for Bob's subsystem, }the preparation of Bob's system given \vicky{by} a choice, $x$, of Alice must be operationally equivalent to that given \vicky{by} any other choice, $x'$, of Alice. In this way, a Bell scenario is viewed as a \emph{remote}-preparation and measurement experiment with preparation equivalences\vicky{, and is therefore, an example of a contextuality scenario}.

\vicky{I}n this \vicky{Letter} we \vicky{use this intuition to define a} mapping between \vicky{these} scenarios and show that the set of quantum correlations in a given two-party Bell scenario is isomorphic to the union of the sets of quantum correlations in an indexed \footnote{Within the family of contextuality scenarios, some scenarios appear multiple times. The indexing of the scenarios avoids multiple Bell correlations (that are equivalent up to relabelling) mapping to the same contextual correlation.} family of contextuality scenarios \footnote{The fact that one Bell scenario maps to a family of contextuality scenarios was not previously acknowledged in the literature.} (see Fig.~\ref{fig:map}). The quantum Bell correlations we consider are those given by the tensor product formalism for potentially infinite dimensional quantum systems, denoted $\mathcal{C}_{qs}$ for quantum spatial correlations. We further show that this mapping is also a bijection between the local/non-signalling correlations and the non-contextual/contextual correlations, respectively, in these scenarios.

\vicky{The map and showing its theory-preserving nature form our first main contribution. Combining these results with the remote-preparation perspective shows that: \emph{if a physical theory predicts the generalised contextual correlations of quantum theory, then that theory is exactly limited to producing the quantum spatial correlations in any two-party Bell scenario, under the no-signalling assumption.} Thus, we demonstrate a characterisation of the set of the quantum spatial correlations in terms of contextuality.} 

Th\vicky{e} connection between two party Bell scenarios and (prepare-and-measure) contextuality scenarios is \vicky{noted} in various works~\cite{LIANG20111,schmid2018contextual,PhysRevA.97.062103}, see also \footnote{Viewing a Bell scenario as a remote-preparation and measurement experiment has also been used to link entanglement and contextuality~\cite{plavala2022contextuality}, as well as non-locality and quantum advantage in \emph{oblivious communication tasks}~\cite{HTMB17,debaanu}.}. \vicky{In these works, t}he relationship is described via examples and the general case is not addressed\vicky{, meaning the statement above was not established}.


\vicky{Furthermore, i}t was previously thought that all contextuality scenarios of a certain kind (in which there are no measurement equivalences and the preparation equivalences comprise various decompositions of one single hypothetical preparation) could be mapped to Bell scenarios in this manner \cite[Sec.~VII]{schmid2018contextual}. However, we find examples of such scenarios in which this mapping is not possible. Of course, this does not rule out an isomorphism in this case but a different map would be required. 

\vicky{In our second main contribution, w}e use our isomorphism to deduce various properties of the quantum set of contextual correlations, including: membership undecidability, the necessity of infinite dimensional quantum systems in realising all quantum correlations, and non-convergence to the quantum set of semidefinite programming (SDP) hierarchies~\cite{tavakoli2021bounding,chaturvedi2021characterising}.

This final result follows from showing that a computable hierarchy of outer approximations converging to the quantum set of contextual correlations would give rise to an algorithm capable of deciding the weak membership problem for the closure $\cC_{qa}$ of $\cC_{qs}$. However, this problem is known to be undecidable as a consequence of the result $\mathrm{MIP}^*=\mathrm{RE}$ \cite{ji2020mip}. This result raises several open questions. To what superset, $\cQ_\infty$, of quantum \vicky{correlations} do the SDP hierarchies in Refs.~\cite{tavakoli2021bounding,chaturvedi2021characterising} converge? What would be the image of $\cQ_\infty$ in the Bell setting under our mapping? A natural candidate could be the set of quantum commuting correlations\vicky{, which generally maps to a strict superset of the quantum contextual correlations under our mapping}. If this is the case, does $\cQ_\infty$ have a physical interpretation in the contextuality setting? Alternatively, the image of $\cQ_\infty$ might provide a new outer approximation of the set $\cC_{qs}$. 

\vicky{In the main text we will describe our map for Bell correlations in which each of Alice's outcomes occurs with non-zero probability. This case encapsulates the central concepts of the map and avoids some technicalities of the general case. In the Appendices we provide a complete description of the map which is used to prove our main results.}

\onecolumngrid\
\begin{center}\
\vspace{-1.5cm}\begin{figure}[b]\
\includegraphics[scale=0.3]{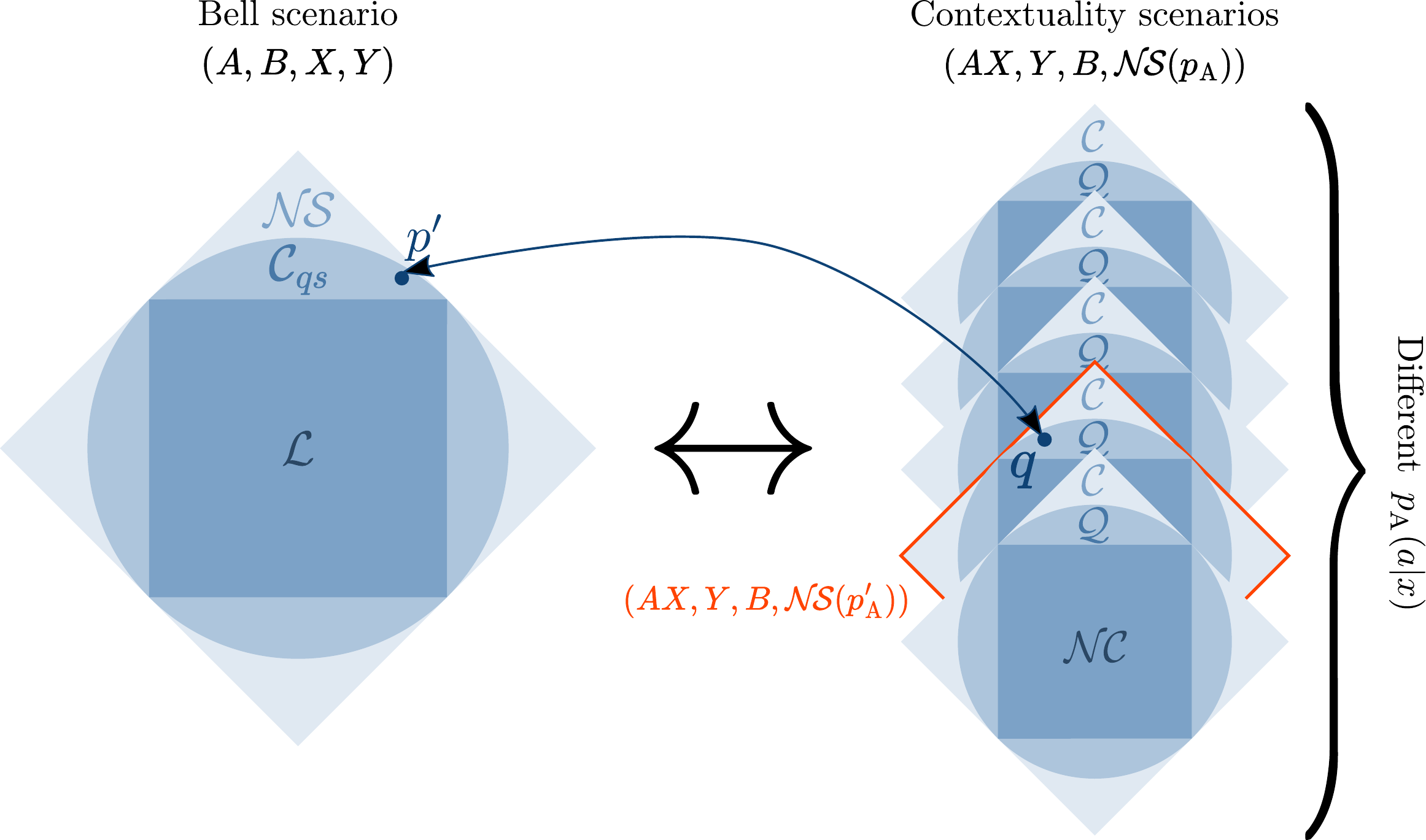}\
\caption{\label{fig:map}A schematic representation of the invertible map between correlations in a Bell scenario and behaviours in a family of contextuality scenarios. \vicky{Here $\cL$, $\cC_{qs}$ and $\cNS$ denote the local, quantum spatial and non-signalling sets of correlations in Bell scenarios while $\cNC$, $\cQ$ and $\cC$ denote the non-contextual, quantum and contextual sets of behaviours in contextuality scenarios (see the main text for more details).}}\
\end{figure}\
\end{center}\
\twocolumngrid\
\prlparagraph{\vicky{Bell scenarios}}
\vicky{A \emph{two-party Bell scenario} comprises two space-like separated experiments. In the first, a party, call her Alice, selects an input from the set $[X]:=\{1,\ldots,X\}$, for some $X\in\N$ and observes an outcome from a set $[A]$, for some $A\in\N$. In the second, another party, call him Bob, similarly selects an input from the set $[Y]$, for some $Y\in\N$ and observes an outcome from a set $[B]$, for some $B\in\N$. The specific scenario can therefore be identified by the tuple of four numbers $(A,B,X,Y)$, indicating the numbers of inputs and outputs for each party. Unless otherwise stated, variables $a,b,x,y$ take values from the sets $[A],[B],[X],[Y]$, throughout.

Given a Bell scenario $(A,B,X,Y)$, a correlation is given by a vector $p\in\R^{ABXY}$, with entries $p(a,b|x,y)$ that specify the probability of Alice and Bob observing outcomes $a$ and $b$ given inputs $x$ and $y$, respectively. In this work we will primarily consider the set, $\cC_{qs}$, of quantum correlations in a Bell scenario using the tensor product formulation and allowing for infinite dimensional quantum systems. A correlation, $p$, is in the quantum set, $\cC_{qs}$, if 

\begin{equation}
p(a,b|x,y)=\Tr\left(M^x_a\otimes N^y_b\rho\right),
\end{equation}
for some positive-operator-valued measures (POVMs---described in the finite-outcome case by a collection of positive semidefinite operators summing to the identity operator) $M^x=\{M^x_a\}_{a}$ and $N^y=\{N^y_b\}_{b}$ on a separable Hilbert spaces $\cH_\tA$ and $\cH_\tB$, respectively, and a density operator (positive semidefinite operator with unit trace) $\rho$ on $\cH_\tA\otimes\cH_\tB$.}

A strict superset of the quantum set is the so-called \textit{no-signalling} set, described by correlations $p$ satisfying the no-signalling constraints
\begin{gather}
\sum_b p(a,b|x,y) = \sum_b p(a,b|x,y') \quad \forall a,x,y,y' \\
\sum_a p(a,b|x,y) = \sum_a p(a,b|x',y) \quad \forall b,y,x,x' .
\end{gather}
A strict subset of the quantum set (considered ``classical'' in Bell scenarios) is the \textit{local set}. A correlation\vicky{,} $p$\vicky{,} is local if there exists a measurable space $(\Lambda,\Sigma)$, a probability measure $\mu:\Sigma\rightarrow[0,1]$, and local probability distributions $l^\tA(a|x,E)$ and $l^\tB(b|y,E)$ satisfying $\sum_a l^\tA(a|x,E)=\sum_b l^\tB(b|y,E)=1$ for all $x$, $y$ and non-empty $E\in\Sigma$, such that
\begin{equation}
p(a,b|x,y)=\int_\Lambda l^\tA(a|x,\lambda)l^\tB(b|y,\lambda)\mathrm{d}\mu(\lambda).
\end{equation} 
\vicky{The relationship between the sets $\cL$, $\cC_{qs}$ and $\cNS$ is depicted on the left-hand side of Fig.~\ref{fig:map}.

\prlparagraph{Contextuality scenarios}
A \emph{contextuality scenario} is an experiment capable of revealing the impossibility of modelling a physical system with a non-contextual ontological model. A key concept in generalised contextuality is operational equivalence, so we will want operational equivalence to appear in our experiment. For our purposes operational equivalence between preparation procedures is sufficient. 

Two preparation procedures\vicky{, $P_1$ and $P_2$,} for a system are operationally equivalent, denoted $P_1\simeq P_2$, in a theory when any outcome of any measurement on the system would occur with the same probability whether the measurement is performed on a system prepared with procedure $P_1$ or $P_2$. 

A \emph{prepare-and-measure contextuality scenario} is an experiment consisting of performing one of $X$ preparation procedures on a system then one of $Y$ measurement procedures. Mixtures of the preparation procedures assigned to each label $x\in[X]$ must satisfy some operational equivalences which are specified by the scenario. These preparation equivalences are of the form:
\begin{equation}\label{eq:prep_equiv}
\sum_x \alpha_x P_x \simeq \sum_x \beta_x P_x,
\end{equation} 
for $\alpha_x\ge0$, $\beta_x\ge0$ such that $\sum_x\alpha_x=\sum_x\beta_x=1$.
For example, a contextuality scenario could have $X=4$ preparations, $P_j$ for $j\in[4]$, that must satisfy $\frac12 P_1 + \frac12 P_2 \simeq \frac12 P_3 + \frac12 P_4$. A valid realisation of this experiment could be to use a qubit system with $P_1$ and $P_2$ being the eigenstates the Pauli--$\operatorname{Z}$ operator, while $P_3$ and $P_4$ are the eigenstates of the Pauli--$\operatorname{X}$ operator.

Generally, a prepare-and-measure contextuality scenario is identified by a tuple $(\cX,\cY,B,\mathcal{OE}_P,\mathcal{OE}_M)$ indicating that it concerns $\cX$ preparations satisfying equivalences $\mathcal{OE}_P$ and $\cY$ measurements each with $B$ outcomes satisfying equivalences $\mathcal{OE}_M$. Since we only consider preparation equivalences we will omit the final element of the tuple.

We are interested in the achievable correlations within a given theory in each contextuality scenario. Each correlation is described by a vector $q\in\R^{XYB}$ with entries given by the probability $q(b|x,y)$ of seeing outcome $b$ after performing measurement $y$ on a system prepared with procedure $x$. We will call these vectors \emph{behaviours} to distinguish them from the correlations in Bell scenarios.} 

A behaviour\vicky{,} $q$\vicky{,} is in the set of \emph{contextual} behaviours (i.e.~behaviours realisable in some contextual theory) if for every equivalence of the form in Eq.~\eqref{eq:prep_equiv} in $\mathcal{OE}_P$ the behaviour satisfies
\begin{equation}
\sum_x \alpha_x q(b|x,y) = \sum_x \beta_x q(b|x,y) \quad \forall b,y.
\end{equation}
\vicky{T}he set of contextual behaviours contains both the sets of quantum and non-contextual behaviours (see below).
A behaviour\vicky{,} $q$\vicky{,} is in the quantum set, $\cQ$, \vicky{of a contextuality scenario $(\cX,\cY,B,\mathcal{OE}_P)$} if 
\begin{equation}
q(b|x,y)=\Tr(N^y_b\rho_x)\,.
\end{equation}
\vicky{for some} POVMs \vicky{$N^y=\{N^y_b\}_{b}$}  on \vicky{a separable Hilbert space} $\cH$ and density operators $\rho_x$ on $\cH$ satisfying $\sum_x \alpha_x \rho_{x} = \sum_x \beta_x \rho_x$ for every equivalence of the form in Eq.~\eqref{eq:prep_equiv} in $\mathcal{OE}_P$.  

A subset of the quantum set (considered ``classical'' in contextuality scenarios) is the \textit{non-contextual set} of behaviours. A behaviour\vicky{,} $q$\vicky{,} is in the non-contextual set if there exists a measurable space $(\Lambda,\Sigma)$, probability measures $\mu_{x}: \Sigma \to [0,1]$ satisfying $\sum_x \alpha_x \mu_{x}(E) = \sum_x \beta_x \mu_x(E)$ for every equivalence relation of the form \eqref{eq:prep_equiv} in $\mathcal{OE}_P$ and so-called \textit{response functions} $\xi_y(b|\cdot)$ for all $b$ and $y$ on $\Lambda$, and $\sum_b\xi(b|E)=1$ for all $E\in\Sigma$, such that 
\begin{equation}
q(b|x,y) = \int_\Lambda\xi_y(b|\lambda)\mathrm{d}\mu_{x}(\lambda)\,.
\end{equation}

\prlparagraph{The map}
We now define an invertible map taking any non-signalling correlation $p$ in a two-party Bell scenario to a behaviour $q$ from one of a family of contextuality scenarios. We will show that this map defines a bijection between (i) non-signalling Bell correlations and contextual behaviours, (ii) quantum Bell correlations and quantum behaviours, and (iii) local Bell correlations and non-contextual behaviours. 

The basic premise is to imagine the Bell experiment  \vicky{$(A,B,\cX,\cY)$} as a prepare-and-measure experiment wherein if Alice inputs $x$ and observes output \vicky{$a$} \vicky{this constitutes a preparation procedure} $P_{a|x}$ \vicky{for Bob's system} on which he will perform a measurement $y$ \vicky{then} observe an outcome $b$. Then, if we impose that Alice cannot signal to Bob\vicky{,} we know that the average preparation Bob receives when Alice inputs any $x$ must be the same as the average preparation he receives when she inputs any other $x'\in[\cX]$. In other words, if the correlation observed in the Bell experiment is $p$ then the preparations $\sum_ap_\tA(a|x)P_{a|x}$ must be equivalent for all $x$, where $p_\tA(a|x)=\sum_bp(a,b|x,y)$ for any $y$ is the marginal distribution of Alice, which is well-defined due to no-signalling.

\vicky{Thus, under the no-signalling assumption, a Bell scenario $(A,B,X,Y)$ implements a contextuality scenario $(AX,Y,B,\cNS(p_\tA))$, where $\cNS(p_\tA)$ denotes the preparation equivalences 
\begin{equation}\label{eq:NSequiv}
\sum_{a}p_\tA(a|1)P_{a|1}\simeq \cdots 
 \simeq\sum_{a}p_\tA(a|\cX)P_{a|\cX}, 
\end{equation}
implied by the no-signalling assumption, which we will encode in the Cartesian product of $X$ vectors in $\R^A$, where the $a$-th element of the $x$-th vector is $\paxax$.

Based on this intuition we define our map for Bell correlations with non-zero marginal distributions for Alice. A correlation $p$ from a Bell scenario $(A,B,X,Y)$ is mapped to a behaviour $q$ in the contextuality scenario $(AX,Y,B,\cNS(p_\tA))$, where 
\begin{equation}\label{eq:q}
q(b|[a|x],y)=\frac{p(a,b|x,y)}{p_\tA(a|x)}.
\end{equation}
Explicitly, our map is 
\begin{equation}
\begin{aligned}
\Gamma:\left[\R^{ABXY},\N^4\right]\to\left[\R^{AXYB},\N^3,\left(\R^A\right)^X\right]&,\\
\left[p,(A,B,X,Y)\right]\mapsto\left[q,(AX,Y,B,\cNS(p_\tA))\right]&,
\end{aligned}
\end{equation}
for $p$ in the non-signalling set of $(A,B,X,Y)$ such that $p_\tA(a|x)\neq0$ where $q$ is defined in Eq.~\eqref{eq:q}.

Notice that the correlations from one Bell scenario are mapped to behaviours from multiple different contextuality scenarios. Each of the contextuality scenarios in the image of a Bell scenario $(A,B,X,Y)$ has $AX$ preparations and $Y$ measurements with $B$ outcomes but the preparation equivalences vary depending on Alice's marginal distribution in the argument correlation. This relationship is depicted in Fig.~\ref{fig:map}.  

We can now define the inverse to our map. Given a contextuality scenario with preparation equivalences satisfying the following criteria, we can always express the equivalences as in Eq.~\eqref{eq:NSequiv}:
\begin{enumerate}[label = (\Roman*)]
\item \label{crit1} comprising a number, $X$, of mixtures each of the same number, $A$, of preparations (since we are considering the case in which $p_\tA(a|x)\neq0$ for all $a$ and $x$) that are all equivalent to one another,
\item and where no preparation appears in more than one mixture.
\end{enumerate} 

The domain of our inverse map will be pairs of a contextuality scenario with such equivalences and a behaviour in that scenario. Explicitly, the inverse of our map is then 
\begin{equation}
\begin{aligned}\label{eq:gaminv}
&\Gamma^{-1}:\left[\R^{AXYB},\N^3,\left(\R^A\right)^X\right]\to\left[\R^{ABXY},\N^4\right],\\
&\left[q,(AX,Y,B,\cNS(p_\tA))\right]\mapsto\left[p,(A,B,X,Y)\right],
\end{aligned}
\end{equation}
for a behaviour $q$ in the contextual set $\cC$ of $(AX,Y,B,\cNS(p_\tA))$ and with $p(a,b|x,y)=p_\tA(a|x)q(b|[a|x],y)$. Note that the $p_\tA(a|x)$ are defined by the coefficients in the preparation equivalences of the contextuality scenario, but end up being equal to the marginals of Alice in the Bell scenario resulting in no conflict of notation.

In Appendix~\ref{fullmap} we extend the map $\Gamma$ to all non-signalling correlations in a given two-party Bell scenario. In this general case, we allow zeroes in the vectors $\cNS(\pax)$ leading to the same contextuality scenario appearing multiple times in the image of the map, but we use the vectors $\cNS(\pax)$ to index the multiple appearances and allow the map to be invertible. Two Bell correlations that are mapped to the same behaviour in two instances of a contextuality scenario are equivalent up to relabelling. 

Under this extension the contextuality scenarios in the image of the map no longer are required to have the same number of preparations in each mixture in the preparation equivalences. That is, criterion~\ref{crit1} for a contextuality scenario to be in the domain of $\Gamma^{-1}$ simply becomes: a number, $X$, of mixtures of preparations that are all equivalent to one another.}

\vicky{We prove our main results about the map} in the Appendices
. \vicky{Namely, in Appendix~\ref{sec:ctb}} we prove that (given a contextuality scenario of the right type) \vicky{$\Gamma^{-1}$} maps every quantum contextual behaviour $q$ to a quantum spatial correlation $p$. We do so by observing that the problem is equivalent to finding a way to steer Bob's system into the assemblage given by the quantum states in the realisation of $q$. The Schr\"{o}dinger--HJW theorem~\cite{kirk06} then provides an explicit construction for realising the quantum correlation $p$. Appendix~\ref{sec:btc} shows that \vicky{$\Gamma$} maps quantum spatial correlations to quantum contextual behaviours. Then Appendices~\ref{sec:ltnc} and~\ref{sec:nstc} treat the cases of local and non-signalling correlations invertibly mapping to non-contextual and contextual behaviours, respectively.

\prlparagraph{Limitations of the map}
In the literature, it is claimed that any contextuality scenario with preparation equivalences given by multiple decompositions of a single hypothetical preparation
\begin{equation}\label{eq:onehyp}
P_\tB\simeq\sum_{a=1}^{Z }p_{a,1}P_a\simeq\sum_{a=1}^{Z }p_{a,2}P_a\simeq\ldots\simeq\sum_{a=1}^{Z }p_{a,\cX}P_a
\end{equation}
is equivalent to a Bell scenario interpreted as a remote prepare-and-measure experiment \cite[Sec.~VII]{schmid2018contextual}. In Appendix~\ref{sec:counterapp}
, we give an example of a sequence of preparation equivalences of the form in Eq.~\eqref{eq:onehyp} that cannot be reduced to a sequence of equivalences $\cNS (p_\tA)$ \vicky{(even when allowing for the coefficients $p_\tA(a|x)$ to be zero)}, i.e.~\vicky{in our example a single preparation appears in multiple different mixtures}. One can still attempt to map such a scenario, $H$, to a Bell scenario, by embedding behaviours $q$ from $H$ into those from a larger scenario $H'$ (yielding a behaviour $q'$), in which each appearance of a preparation that appears multiple times in $\mathcal{OE}_P$ is treated as a distinct preparation. The resulting sequence of equivalences $\mathcal{OE}'_P$ is of the form $\cNS (p_\tA)$. However, we show via an explicit example that this embedding can map a contextual behaviour $q_c$ in $H$ to a non-contextual behaviour $q'_c$ in $H'$. Thus, using this embedding to connect $H$ to a Bell scenario leads to a contextual behaviour ($q_c$) being mapped to a local correlation [through the embedding $q'_c$ and then Eq.~\eqref{eq:gaminv}]. Therefore, the connection between non-contextuality and locality would be lost by composing this embedding and our map \vicky{$\Gamma$}.

\prlparagraph{The quantum set in contextuality scenarios}
Using the connection we have made between the sets of quantum behaviours in contextuality scenarios and quantum correlations in Bell scenarios, we can transfer various results about quantum non-locality to contextuality. \vicky{Our main results about the quantum contextual set are given in the following four corollaries with proofs} in Appendices~\ref{sec:proof_membership}--\ref{sec:proof_compute}.  
\begin{cor}\label{cor:membership}
The membership problem for the set of quantum behaviours in a contextuality scenario is undecidable.
\end{cor} 
\begin{cor}\label{cor:finite-dim}
The set of behaviours deriving from finite-dimensional quantum systems in contextuality scenarios is a strict subset of its infinite-dimensional counterpart.
\end{cor}
\begin{cor}\label{cor:nclosed}
In general, the set of behaviours in a contextuality scenario is not closed.
\end{cor}
\begin{cor}\label{cor:compute} No hierarchy of SDPs converges to the quantum contextual set $\cQ$ or its closure $\overline{\cQ}$ for all contextuality scenarios.
\end{cor}
Note that the SDP hierarchy in Corollary \ref{cor:compute} could be replaced by any algorithm capable of verifying that a behaviour is $\varepsilon$ away from $\cQ$ (in $\ell_1$ distance) for all $\varepsilon>0$.

\prlparagraph{Conclusion and outlook}
We constructed an isomorphism between the set of quantum spatial correlations and \vicky{the} set of quantum contextual behaviours \vicky{from an indexed family of contextual scenarios}. This map allows us to \vicky{characterise quantum non-locality in terms of quantum contextuality,} translate results from Bell non-locality to generalised contextuality, and also raises questions \vicky{about} the limits of SDP hierarchies in contextuality scenarios (see the Introduction). \vicky{A} natural future research direction \vicky{would be} to investigate whether other results from Bell non-locality, such as self-testing \cite{SB20} and device-independent quantum key distribution \cite{DIQKD}, have analogs in contextuality scenarios that can be found via our construction. Lastly, one might attempt to generalise our map to multipartite Bell scenarios. One such natural generalisation remains a bijection between local and non-contextual, and between non-signalling and contextual sets in multipartite Bell scenarios. However, whether this map also preserve\vicky{s} quantumness remains unknown, with the existence of post-quantum steering~\cite{postquantum} posing an obstacle to generalising our argument.

\prlparagraph{Acknowledgements}
We thank Miguel Navascu\'es for details of the proof of the Schr\"odinger--HJW theorem in the infinite-dimensional case, and Anubhav Chaturvedi, Luke Mortimer and Gabriel Senno for fruitful discussions. This  project  has  received  funding  from  the  European  Union's  Horizon~2020  research and innovation programme under the Marie Sk\l{}odowska-Curie grant agreement No.~754510, the Government of Spain (FIS2020-TRANQI, Severo Ochoa CEX2019-000910-S), Fundaci\'o Cellex, Fundaci\'o Mir-Puig and Generalitat de Catalunya (CERCA, AGAUR SGR 1381).

{\color{white}\cite{navascues2012physical}}

\bibliography{context.bib}
\bibliographystyle{unsrt}

\appendix

\onecolumngrid

\vicky{\section{The map}\label{fullmap}
To specify the map in fully generality it will be useful for us to be able to describe Bell and contextuality scenarios in which the different measurements have different numbers of outcomes. Thus, we redefine our tuples that denote the scenarios as follows.

We use a tuple $(\bA,\bB,X,Y)$ to denote a two-party Bell scenario in which Alice (Bob) has $X$ ($Y$) inputs and given an input $x\in[X]$ ($y\in[Y]$) she (he) can obtain one of $A_x$ ($B_y$) possible outcomes, where $A_x$ ($B_y$) are the entries of the $X$-tuple  $\bA$ ($Y$-tuple $\bB$). We use the notation $\norm{\bA}=\sum_x A_x$ and $\norm{\bB}=\sum_y B_y$. 

Given a Bell scenario $(\bA,\bB,\cX,\cY)$, a correlation is given by a vector $p\in\R^{\norm{\bA}\norm{\bB}}$, with entries $p(a,b|x,y)$. A correlation $p$ is in the quantum set, $\cC_{qs}$, if there exist separable Hilbert spaces $\cH_\tA$ and $\cH_\tB$, positive-operator-valued measures (POVMs) $M^x=\{M^x_a\}_{a\in [A_x]}$ for all $x\in[\cX]$ on $\cH_\tA$ and $N^y=\{N^y_b\}_{b\in [B_y]}$ for all $y\in[\cY]$ on $\cH_\tB$, and a density operator (positive semidefinite operator with unit trace) $\rho$ on $\cH_\tA\otimes\cH_\tB$ such that
\begin{equation}
p(a,b|x,y)=\Tr\left(M^x_a\otimes N^y_b\rho\right).
\end{equation}
A correlation $p$ is in the no-signalling set if it satisfies the no-signalling constraints
\begin{gather}
\sum_b p(a,b|x,y) = \sum_b p(a,b|x,y') \quad \forall a,x,y,y' \\
\sum_a p(a,b|x,y) = \sum_a p(a,b|x',y) \quad \forall b,y,x,x' .
\end{gather}
A correlation $p$ is local if there exists a measurable space $(\Lambda,\Sigma)$, a probability measure $\mu:\Sigma\rightarrow[0,1]$, and local probability distributions $l^\tA(a|x,E)$ and $l^\tB(b|y,E)$ satisfying $\sum_a l^\tA(a|x,E)=\sum_b l^\tB(b|y,E)=1$ for all $x\in[\cX]$, $y\in[\cY]$ and non-empty $E\in\Sigma$, such that
\begin{equation}
p(a,b|x,y)=\int_\Lambda l^\tA(a|x,\lambda)l^\tB(b|y,\lambda)\mathrm{d}\mu(\lambda).
\end{equation} 

We identify a prepare-and-measure contextuality scenario with $\cX$ preparations satisfying equivalences $\mathcal{OE}_P$ and $\cY$ measurements where measurement $y \in [Y]$ has $B_y$ outcomes by the tuple $(\cX,\cY,\bB,\mathcal{OE}_P)$, where $\bB$ is a $\cY$-tuple with $y$th element $B_y$. 

A behaviour $q$ is in the set of \emph{contextual} behaviours if for every equivalence of the form 

\begin{equation}\label{eq:prep_equiv}
\sum_x \alpha_x P_x \simeq \sum_x \beta_x P_x,
\end{equation} in $\mathcal{OE}_P$ the behaviour satisfies
\begin{equation}
\sum_x \alpha_x q(b|x,y) = \sum_x \beta_x q(b|x,y) \quad \forall b,y.
\end{equation}

A behaviour $q$ is in the quantum set, $\cQ$, if there exists a separable Hilbert space $\cH$, POVMs $\{N^y_b\}_{b\in [B_y]}$ for $y\in[\cY]$ on $\cH$ satisfying $\mathcal{OE}_M$ and density operators $\rho_x$ on $\cH$ satisfying $\sum_x \alpha_x \rho_{x} = \sum_x \beta_x \rho_x$ for every equivalence of the form in Eq.~\eqref{eq:prep_equiv} in $\mathcal{OE}_P$ such that 
\begin{equation}
q(b|x,y)=\Tr(N^y_b\rho_x)\,.
\end{equation}

A behaviour $q$ is in the non-contextual set if there exists a measurable space $(\Lambda,\Sigma)$, probability measures $\mu_{x}: \Sigma \to [0,1]$ for all $x\in[\cX]$ satisfying $\sum_x \alpha_x \mu_{x}(E) = \sum_x \beta_x \mu_x(E)$ for every equivalence relation of the form \eqref{eq:prep_equiv} in $\mathcal{OE}_P$ and so-called \textit{response functions} $\xi_y(b|\cdot)$ for all $b\in[B_y]$ and $y\in[\cY]$ on $\Lambda$, and $\sum_b\xi(b|E)=1$ for all $E\in\Sigma$, such that 
\begin{equation}
q(b|x,y) = \int_\Lambda\xi_y(b|\lambda)\mathrm{d}\mu_{x}(\lambda)\,.
\end{equation}

Now we can define the map $\Gamma$ on the full set of non-signalling correlations. Notice that if there is an outcome $a|x$ that never occurs for Alice in the correlation $p$ then $p_\tA(a|x)=0$ and the corresponding preparation $P_{a|x}$ of Bob's system does not appear in the preparation equivalences $\cNS(p_\tA)$ (since it would have coefficient zero). Consequently, the choice of preparation $P_{a|x}$ would be unconstrained in both contextual and non-contextual theories. As a result the correlation $p$ could be mapped to one of many possibilities. Since we wish to map each non-signalling correlation to a single behaviour in a contextual scenario, we will not include these unconstrained preparations in the contextuality scenario.

A correlation $p$ in a Bell scenario $(\bA,\bB,\cX,\cY)$ is mapped to a behaviour $q$ in the contextuality scenario $(\norm{\cA^p},\cY,\bB,\cNS (p_\tA))$, where $\norm{\cA^{p}}=\sum_{x\in[\cX]}\abs{\cA^p_x}$ and $\cA^p_x = \{a : p_\text{A}(a|x)>0\}$ for all $x\in[\cX]$, and where the preparation equivalences are given by 
\begin{equation}\label{eq:nsequiv}
  \sum_{a\in[A_1]}p_\tA(a|1)P_{a|1}\simeq 
 \cdots \simeq\sum_{a\in[A_\cX]}p_\tA(a|\cX)P_{a|\cX},
\end{equation}
which we encode as a vector $\cNS (p_\tA)$ in the Cartesian product $\R^{A_1}\times\cdots\times\R^{A_X}$ where the $x$-th vector has $a$-th element $\paxax$.
If a preparation $P_{a|x}$ has coefficient zero we will say $P_{a|x}$ does not appear in $\cNS(p_\tA)$.

The mapping is given by
\begin{equation}\label{eq:Gam}
\Gamma:[p,(\bA,\bB,\cX,\cY)]\mapsto[q,\left(\norm{\cA^p},\cY,\bB,\cNS (p_\tA)\right)]
\end{equation}
where 
\begin{equation}\label{eq:btc}
q(b|[a|x],y) = \frac{ p(a,b|x,y) }{ p_\tA(a|x) }\,
\end{equation}
for $a\in\cA^p_x$, $b\in[B_y]$, $x\in[X]$ and $y\in[Y]$.

For the inverse mapping, let ${A}_x\in\N$ for all $x\in[X]$ and some $X\in\N$, then consider some $\hat{p}_\tA(a|x)\ge0$ for all $a\in[A_x]$ and $x\in[X]$ such that $\sum_{a\in{A}_x}\hat{p}_\tA(a|x)=1$ for all $x\in[X]$. Then, let $\cNS (\tpa)\in\R^{{A}_1}\times\cdots\times\R^{{A}_X}$ be such that the $a$-th element of the $x$-th vector is $\hat{p}_\tA(a|x)$. We can now define a contextuality scenario $(Z,\cY,\bB,\cNS (\tpa))$ where the number of preparations, $Z$, is the number of non-zero elements in all the vectors of $\cNS (\tpa)$. Here $\cNS (\tpa)$ encodes preparation equivalences as described in Eq.~\eqref{eq:nsequiv}. The inverse mapping $\Gamma^{-1}$ takes a behaviour $q$ in this contextuality scenario to a correlation $p$ in the Bell scenario $(\bA=({A}_1,\ldots,{A}_X),\bB,\cX,\cY)$ and is given by 
\begin{equation}
\Gamma^{-1}:[q,\big(Z ,\cY,\bB, \cNS (\tpa)\big)]
\mapsto[p,(\bA,\bB,\cX,\cY)]
\end{equation}
where 
\begin{equation} \label{eq:ctb}
p(a,b|x,y)=\tpa(a|x)q(b|[a|x],y)
\end{equation}
for $b\in[B_y]$, $x\in[\cX]$ and $y\in[\cY]$. When the scenarios are already specified\vicky{,} we will refer to the $q$ in Eq.~\eqref{eq:btc} as $\Gamma(p)$ and similarly, the $p$ in Eq.~\eqref{eq:ctb} as $\Gamma^{-1}(q)$.

If we are simply given a behaviour in a contextuality scenario of the right type\footnote{A scenario with no measurement equivalences and preparation equivalences given by multiple decompositions of one hypothetical preparation in which each preparation only appears once.} we have a choice of Bell scenario into which to map. This choice stems from being able to consider a correlation in a Bell scenario as a correlation from a larger scenario where some outcomes never occur. For example, consider a contextuality scenario with five preparations $Q_1,\ldots, Q_5$ and preparation equivalences $\frac12 Q_1+\frac12 Q_2\simeq\frac14 Q_3 +\frac34 Q_4\simeq Q_5$, and two measurements with two outcomes each. The simplest choice would be to map to a Bell scenario with three measurements for Alice, the first two having two outcomes each and the third having one (trivial) outcome. This choice corresponds to thinking of the preparations as $Q_1=P_{1|1}$, $Q_2=P_{2|1}$, $Q_3=P_{1|2}$, $Q_4=P_{2|2}$ and $Q_5=P_{1|3}$ which results from choosing $\cNS(\pax)=(1/2,1/2)\times(1/4,3/4)\times(1)$.

However, an equally valid choice is to take $\cNS(\pax)=(0,1/2,0,0,1/2)\times(1/4,3/4,0)\times(0,1)$ which means we would think of the preparations as $Q_1=P_{2|1}$, $Q_2=P_{5|1}$, $Q_3=P_{1|2}$, $Q_4=P_{2|2}$ and $Q_5=P_{2|3}$. The resulting Bell scenario is $\big((5,3,2),(2,2),3,2\big)$. In this case the behaviours from this contextuality scenario will map to the correlations in the Bell scenario with marginals given by the coefficients in the equivalences, e.g.~$p_\tA(1|2)=1/4$, and in which all Alice's outcomes that do not have a corresponding preparation never occur, e.g.~$p_\tA(1|1)=0$. It follows from the results of the present manuscript that if the image of a behaviour in the first Bell scenario is a local, quantum or non-signalling correlation, then the image in the second Bell scenario will also be local, quantum or non-signalling, respectively.}
\section{Quantum case---Contextuality to Bell}\label{sec:ctb}

In this section we will show that \vicky{$\Gamma^{-1}$} always takes a quantum behaviour $q$ \vicky{in a scenario $(Z ,\cY,\bB, \cNS (\tpa))$} to a quantum correlation $p$. In the case of the simplest Bell scenario some similar arguments are described in Refs.~\cite{leifer2013maximally, pusey2018robust}. Consider a contextuality scenario $(Z ,\cY,\bB,\cNS (\tpa))$ \vicky{as described in the previous section, where $\tpa(a|x)\ge0$ for $a\in[A_x]$ and $x\in[X]$ and $\cNS (\tpa)\in\R^{{A}_1}\times\cdots\times\R^{{A}_X}$ has $a$-th element in the $x$-th vector $\tpa(a|x)$ and} where $Z$ is the number of \vicky{strictly positive elements of the vectors} in $\cNS (\tpa)$. \vicky{We will also denote by $\hat{\cA}_x$ the subset of $[A_x]$ such that $\tpa(a|x)>0$.} Let $q$ be a quantum behaviour in this scenario with a realisation\vicky{, that is}
\begin{equation}
q(b|[a|x],y)=\Tr(\rho_{a|x}N^y_b)\,,
\end{equation}
for \vicky{some} density operators $\rho_{a|x}$ \vicky{for all $a\in\hat{\cA}_x$ and $x\in[X]$}, and \vicky{some} POVMs $\{N^y_b\}_b$ \vicky{for all $b\in[B_y]$ and $y\in[Y]$} on a Hilbert space $\cH$, where $\sum_{a\in\hat{\cA}_x}\tpa(a|x)\rho_{a|x}=\rho_\tB$  for all $x\in[\cX]$ and some density operator $\rho_\tB$. 

 We will show that the correlation 
\begin{equation}
p(a,b|x,y)=\tpa(a|x)q(b|[a|x],y)
\end{equation}
for \vicky{$a\in[A_x]$,} $b\in[B_y]$, $x\in[\cX]$ and $y\in[\cY]$ in the Bell scenario $(\bA,\bB,\cX,\cY)$ has a quantum realisation.

Now, we want to find POVMs $M^x=\{M^x_a\}_{a\in[A_x]}$ on a Hilbert space $\cH_\tA$ and a density operator $\rho$ on $\cH_\tA\otimes\cH_\tB$ such that if Alice measures the POVM $M^x$ on system \vicky{$\tA$} then with probability $\tpa(a|x)$ (from our operational equivalences \vicky{$\cNS (\tpa)$}) she sees outcome $a$ and Bob is left with the state $\rho_{a|x}$ (from our quantum realisation in the contextuality scenario) for $a\in\hat{\cA}_x$ \vicky{and outcome $M^x_a$ never occurs for $a\in[A_x]\backslash\hat{\cA}_x$}. In other words, we are looking for a way for Alice to steer Bob's system into the assemblage given by the states $\rho_{a|x}$. 

Mathematically, we want 
\begin{equation}\label{eq:verify}
\Tr_\tA\left(M^x_a\otimes\I\rho\right)=\begin{cases}
\tpa(a|x)\rho_{a|x}\quad&\text{ if }a\in\hat{\cA}_x\\
0\quad&\text{otherwise}\,,
\end{cases}
\end{equation} 
since then, if Alice measures \vicky{$M^x$} on system A and Bob measures $N^y$ (also from the quantum realisation in the contextuality scenario) on system B when the system AB is in state $\rho$ we find
\begin{equation}
p(a,b|x,y)=\Tr(M^x_a\otimes N^y_b \rho)=
\begin{cases}
\tpa(a|x)\Tr(N^y_b\rho_{a|x})=\tpa(a|x)q(b|[a|x],y)\quad&\text{ if }a\in\hat{\cA}_x\\0\vicky{=\tpa(a|x)q(b|[a|x],y)}\quad&\text{otherwise}\,,
\end{cases}
\end{equation}
%
and we have a quantum realisation for our Bell correlation. Our construction of \vicky{$M^x$} and $\rho$ is based on the Schr\"odinger--HJW theorem~\cite{kirk06}, in particular, on the infinite dimensional argument given by Navascu\'es~\vicky{et al.}~\cite [Lemma 4]{navascues2012physical}.

Since $\rho_\tB$ is a density operator, it has a spectral decomposition given by a countable sum
\begin{equation}
\rho_\tB=\sum_{n}\lambda_n\ket{n}\bra{n},
\end{equation}
where $\left\{\ket{n} : n\in\mathbb{N} \right\}$ is a set of orthonormal vectors of $\cH$, and the positive eigenvalues $\lambda_n$ of $\rho_\tB$ satisfy $\sum_{n}\lambda_n=1$. Note that we have excluded any zero eigenvalues from this decomposition. Let $\cH_\tB$ be the support of $\rho_\tB$ and $\Pi_\tB=\sum_{n}\ket{n}\bra{n}$ be the projection onto this closed subspace of $\cH$. Thus, we have that $\{\ket{n} : n\in\mathbb{N}\}$ forms an orthonormal basis for $\cH_\tB$ and $\Pi_\tB\rho_\tB\Pi_\tB=\rho_\tB$. Furthermore,  we have $\Pi_\tB\rho_{a|x}\Pi_\tB=\rho_{a|x}$ for all $a\in\hat{\cA}_x$ and $x\in[\cX]$ since we have $1=\Tr(\Pi_\tB\rho_\tB\Pi_\tB)=\sum_{a'\in\hat{\cA}_x}\tpa(a'|x)\Tr\left(\Pi_\tB\rho_{a'|x}\Pi_\tB\right)$ which implies $\Tr(\Pi_\tB\rho_{a|x}\Pi_\tB)=1$ and therefore $\Pi_\tB\rho_{a|x}\Pi_\tB=\rho_{a|x}$ for all $a\in\hat{\cA}_x$.

First, we define the state for the realisation of the quantum behaviour $p$. Let $\cH_\tA=\cH_\tB$, then the state is given by $\ket{\Psi}=\sum_{n}\sqrt{\lambda_n}\ket{n}\ket{n}\in\cH_\tA\otimes\cH_\tB$. This series converges since $\sum_{n}\norm{\sqrt{\lambda_n}\ket{n}\ket{n}}^2=\sum_{n}\lambda_n=1$.

Second, we define the effects $M^x_a$ for $a\in\hat{\cA}_x$ of the POVM in the quantum realisation of $p(a,b|x,y)$. 
To do so, we define the operators
\begin{equation} 
M_{a,K}^x=\left[\sum_{m=1}^K\frac{1}{\sqrt{\lambda_m}}\ket{m}\bra{m}\right]\tpa(a|x)\rho_{a|x}^T\left[\sum_{n=1}^K\frac{1}{\sqrt{\lambda_n}}\ket{n}\bra{n}\right]
\end{equation}
on $\cH_\tA$, where the transpose is taken in the eigenbasis of $\rho_\tB$, and $K\in\mathbb{N}$. From now on, we assume that $\cH_\tA$ is infinite dimensional, since the finite-dimensional case is simpler. We will show that the operator $M^x_a$ on $\cH_\tA$ given by $M^x_a\ket{\psi}=\lim_{K\to\infty}M_{a,K}^x\ket{\psi}$ is a bounded linear operator on $\cH_\tA$ for all $a\in\hat{\cA}_x$ and $x\in[\cX]$, and these operators will form the effects of the POVMs in the quantum realisation. 

We have $M_{a,K}^x\geq 0$ for all $K\in\mathbb{N}$ since $M_{a,K}^x=S^\dagger \tpa(a|x)\rho_{a|x}^TS$, where $S=\sum_{m=1}^K\frac{1}{\sqrt{\lambda_m}}\ket{m}\bra{m}$, and $\tpa(a|x)\rho_{a|x}^T$ is positive semidefinite (since transposition is a positive map). Further, we have that $M_{a,K}^x\leq \I$, since
\begin{equation}
\begin{aligned}
\sum_aM_{a,K}^x&=\left[\sum_{m=1}^K\frac{1}{\sqrt{\lambda_m}}\ket{m}\bra{m}\right]\rho^T_\tB\left[\sum_{n=1}^K\frac{1}{\sqrt{\lambda_n}}\ket{n}\bra{n}\right]
\\&=\left[\sum_{m=1}^K\frac{1}{\sqrt{\lambda_m}}\ket{m}\bra{m}\right]\sum_{j=1}^\infty\lambda_j\ket{j}\bra{j}\left[\sum_{n=1}^K\frac{1}{\sqrt{\lambda_n}}\ket{n}\bra{n}\right]
\\&=\sum_{m=1}^K\ket{m}\bra{m}\leq\I.
\end{aligned}
\end{equation}
Thus, $M_{a,K}^x$ is a positive semidefinite bounded linear operator on $\cH_\tA$ for all $K\in\mathbb{N}$, $a\in\hat{\cA}_x$, and $x\in[\cX]$. 

Consider the vector subspace $\cF$ of $\cH_\tA$ consisting of finite linear combinations of the eigenvectors of $\rho_\tB$, i.e.~vectors of the form $\sum_{n=1}^L c_n\ket{n}$, for some $L\in\mathbb{N}$.
\vicky{The limit $\lim_{K\to\infty}M_{a,K}^x\ket{\psi}$ converges in $\cH_\tA$ f}or every element $\ket{\psi}$ of the subspace $\cF$\vicky{, since} there exists some $K_\psi \in\mathbb{N}$ such that \vicky{$\lim_{K\to\infty}M_{a,K}^x\ket{\psi}=M_{a,K_\psi}^x\ket{\psi}=M^x_a\ket{\psi}$}. Thus, we find that on $\cF$, $\norm{M^x_a\ket{\psi}} = \norm{M_{a,K_\psi}^x\ket{\psi}}\leq \norm{\ket{\psi}}$ and, hence, $M^x_a$ is a bounded linear operator on $\mathcal{F}$. Since $\cF$ is a dense subspace of $\cH_\tA$, we have that $M^x_a$ has a unique linear extension to $\cH_\tA$. This extension is defined as follows: given a vector $\ket{\psi}\in\cH_\tA$ let $(\ket{\psi_j})_{j\in\mathbb{N}}$ be a sequence in $\cF$ such that $(\ket{\psi_j})_j\to\ket{\psi}$ as $j\to\infty$. Then we define $M^x_a\ket{\psi}=\lim_{j\to\infty}M^x_a\ket{\psi_j}$. 
This limit exists since the sequence $(M^x_a\ket{\psi_j})_{j}$ is Cauchy which can be seen by the following argument. The sequence $(\ket{\psi_j})$ is Cauchy, therefore for every $\epsilon>0$ there exists $N\in\mathbb{N}$ such that $\norm{\ket{\psi_m}-\ket{\psi_n}}<\epsilon$ for all $m,n>N$. Thus, we have that $\norm{M^x_a\ket{\psi_m}-M^x_a\ket{\psi_n}}\leq\norm{\ket{\psi_m}-\ket{\psi_n}}<\epsilon$ since $\ket{\psi_m}-\ket{\psi_n}\in\mathcal{F}$.

Finally, we set $M^x_a=0$ for the remaining values of $a$, i.e.~for \vicky{$a\in[A_x]\backslash\hat{\cA}_x$} and verify that Eq.~\eqref{eq:verify} holds, that is, we have a quantum realisation of our correlation $p(a,b|x,y)$, given by the POVMs $M^x=\{M^x_a\}$ on $\mathcal{H}_\tA$ for Alice and $N^y=\{N^y_b\}$ on $\cH_\tB$ for Bob and the quantum state $\ket{\Psi}=\sum_{n}\sqrt{\lambda_n}\ket{n}\ket{n}\in\cH_\tA\otimes\cH_\tB$. Clearly, for $a\in[\cA_x]\backslash\hat{\cA}_x$ we have $\Tr_\tA(M^x_a\otimes\I\ket{\Psi}\bra{\Psi})=0$. Then, for $a\in\hat{\cA}_x$ we have
\begin{equation}
\begin{aligned}
\Tr_\tA(M^x_a\otimes\I\ket{\Psi}\bra{\Psi})&=\sum_{j,m,n}\sqrt{\lambda_m\lambda_j}\bra{j}\otimes\I\left( M^x_a\otimes\vicky{\I}\ket{m}\bra{n}\otimes\ket{m}\bra{n}\right)\ket{j}\otimes\I\\
&=\sum_{j,m}\sqrt{\lambda_m\lambda_j}\bra{j}\lim_{K\to\infty}M^x_{a,K}\ket{m}\ket{m}\bra{j}\\
&=\sum_{j,m}\sqrt{\lambda_m\lambda_j}\sum_{k=1}^\infty\frac{p_\tA(a|x)}{\sqrt{\lambda_m\lambda_k}}\langle j|k\rangle\bra{n}\rho^T_{a|x}\ket{m}\ket{m}\bra{j}\\
&=\sum_{j,m}\tpa(a|x)\bra{m}\rho_{a|x}\ket{j}\ket{m}\bra{j}=\tpa(a|x)\rho_{a|x}\,.
\end{aligned}
\end{equation}

\section{Quantum case---Bell to contextuality}\label{sec:btc}

Conversely, consider a quantum correlation, $p$, from a bipartite Bell scenario $(\bA,\bB,\cX,\cY)$ given by
\begin{equation}
p(a,b|x,y)=\Tr(M^x_a\otimes N^y_b\rho)\,,
\end{equation}
for some POVMs $M^x=\{M^x_a\}$ on a Hilbert space $\cH_\tA$ and $N^y=\{N^y_b\}$ on a Hilbert space $\cH_\tB$ and a density operator $\rho$ on $\cH_\tA\otimes\cH_\tB$. Denote Alice's marginal probabilities by $p_\tA(a|x)=\Tr(M^x_a\otimes\I\rho)$ and Bob's reduced states after outcome $a$ of measurement $x$ of Alice with $p_\tA(a|x)\neq 0$ by
\begin{equation}
\rho_{a|x}=\frac{\Tr_\tA(M^x_a\otimes\I\rho)}{p_\tA(a|x)}\,. 
\end{equation}

\vicky{Recalling that we denote the subset of $[A_x]$ such that $\paxax>0$ by $\cA^p_x$, w}e have that $\sum_{a\in\cA^p_x}p_\tA(a|x)\rho_{a|x}=\Tr_\tA(\rho)$ for all $x\in[\cX]$, therefore the density operators $\rho_{a|x}$ satisfy the equivalences $\cNS (p_\tA)$ \vicky{given in Eq.~\eqref{eq:nsequiv}}. Taking $\rho_{a|x}$ as preparation $P_{a|x}$ and $N^y$ as the $y$-th measurement in the contextuality scenario $\left(\norm{\cA^p},\cY,\bB,\cNS (p_\tA)\right)$ results in the behaviour
\begin{equation}
q(b|[a|x],y)=\Tr(N^y_b\rho_{a|x})=\frac{\Tr(M^x_a\otimes N^y_b\rho)}{p_\tA(a|x)}= \frac{p(a,b|x,y)}{p_\tA(a|x)} \,.
\end{equation}
Thus, $q$ is a quantum behaviour in the contextuality scenario $\left(\norm{\cA^p},\cY,\bB,\cNS (p_\tA)\right)$.

\section{Local and non-contextual case}\label{sec:ltnc}
Let $p(a,b|x,y)$ be a local correlation in a Bell scenario $(\bA,\bB,\cX,\cY)$. Then, there exists a measurable space $(\Lambda,\Sigma)$, a probability measure $\mu:\Sigma\rightarrow[0,1]$ and local probability distributions $l^\tA(a|x,E)$ and $l^\tB(b|y,E)$ satisfying $\sum_al^\tA(a|x,E)=\sum_bl^\tB(b|y,E)=1$ for all $x\in[\cX]$, $y\in[\cY]$ and non-empty $E\in\Sigma$ such that 
\begin{equation}
p(a,b|x,y)=\int_\Lambda l^\tA(a|x,\lambda)l^\tB(b|y,\lambda)\mathrm{d}\mu(\lambda).
\end{equation} 
We now construct a non-contextual ontological model that yields the behaviour 
\begin{equation}
q(b|[a|x],y)= \frac{ p(a,b|x,y) }{ p_\tA(a|x) }
\end{equation}
in the contextuality scenario $\left(\norm{\cA^p},\cY,\bB,\cNS (p_\tA)\right)$. Note that the marginals of Alice, $p_\tA(a|x)=\sum_b p(a,b|x,y)$, can now be expressed as $p_\tA(a|x)=\int_\Lambda l^\tA(a|x,\lambda)\mathrm{d}\mu(\lambda)$. 

We select the ontic state space $(\Lambda,\Sigma)$ and each preparation $P_{a|x}$ for $a\in\cA^p_x$ and $x\in[\cX]$ is given by the measure
\begin{equation}
\mu_{a|x}(E)= \frac{l^\tA(a|x,E)\mu(E)}{p_\tA(a|x)}
\end{equation}
on $\Lambda$. These are indeed probability measures on $\Lambda$, since
\begin{equation}
\int_{\Lambda}\mathrm{d}\mu_{a|x}(\lambda)=\frac{1}{p_\tA(a|x)}\int_{\Lambda}l^\tA(a|x,\lambda)\mathrm{d}\mu(\lambda)=1\,.
\end{equation}

Furthermore, the measures $\mu_{a|x}$ satisfy the operational equivalences $\cNS (p_\tA)$, since
\begin{equation}
\sum_ap_\tA(a|x)\mu_{a|x}(E)=\mu(E),
\end{equation}
for all $E\in\Sigma$ and $x\in[\cX]$. 

The response function for each measurement $M^y$ for $y\in[\cY]$ is given by $\xi_y(b|E)=l^\tB(b|y,E)$. Now, for all $a\in\cA^p_x$ and $x\in[\cX]$ we find that
\begin{equation}
q(b|[a|x],y)=\int_{\Lambda}\xi_y(b|\lambda)\mathrm{d}\mu_{a|x}(\lambda)=\int_{\Lambda} l^\tB(b|y,\lambda)\frac{l^\tA(a|x,\lambda)}{p_\tA(a|x)}\mathrm{d}\mu(\lambda)=\frac{p(a,b|x,y)}{p_\tA(a|x)}\,.
\end{equation} 

Conversely, consider a non-contextual behaviour $q(b|[a|x],y)$ in a contextuality scenario $\left(Z ,\cY,\bB,\cNS (\tpa)\right)$ \vicky{where $\tpa(a|x)\ge0$ for $a\in[A_x]$ and $x\in[X]$ and $\cNS (\tpa)\in\R^{{A}_1}\times\cdots\times\R^{{A}_X}$ has $a$-th element in the $x$-th vector $\tpa(a|x)$ and where $Z$ is the number of strictly positive elements of the vectors in $\cNS (\tpa)$. We will also denote by $\hat{\cA}_x$ the subset of $[A_x]$ such that $\tpa(a|x)>0$.}

Then, there exists a measurable space $(\Lambda,\Sigma)$, probability measures $\mu_{a|x}$ for all $a\in\hat{\cA}_x$ and $x\in[\cX]$ and $\xi_y(b|\cdot)$ for all $b\in[B_y]$ and $y\in[\cY]$ on $\Lambda$ satisfying $\sum_a\tpa(a|x)\mu_{a|x}(E)=\mu(E)$ for all $x\in[\cX]$ and $\sum_b\xi(b|E)=1$ for all $E\in\Sigma$ such that
\begin{equation}
\int_\Lambda\xi_y(b|\lambda)\mathrm{d}\mu_{a|x}(\lambda)=q(b|[a|x],y)\,.
\end{equation}
Note that we can assume that $\mu(E)>0$ for all non-empty $E\in\Sigma$. 

We will construct a local hidden variable model for the correlation
\begin{equation}
p(a,b|x,y)=\begin{cases}\tpa(a|x)q(b|[a|x],y)\quad&\text{ if }a\in\hat{\cA}_x\\
0\quad&\text{otherwise}\,,\end{cases}
\end{equation}
for $b\in[B_y]$, $x\in[\cX]$ and $y\in[\cY]$. Let
\begin{equation}
l^\tA(a|x,E)=\begin{cases}
\frac{\tpa(a|x)\mu_{a|x}(E)}{\mu(E)}\qquad &\text{if }a\in\hat{\cA}_x\text{ and }E\text{ non-empty } \\
0&\text{otherwise,}
\end{cases}
\end{equation}
and $l^\tB(b|y,E)=\xi_y(b|E)$. Then, for $a\in\hat{\cA}_x$, we find
\begin{equation}
\begin{aligned}
p(a,b|x,y)&=\int_{\Lambda}l^\tA(a|x,\lambda)l^\tB(b|y,\lambda)\mathrm{d}\mu(\lambda)\\
&=\int_{\Lambda}\tpa(a|x)\xi_y(b|\lambda)\mathrm{d}\mu_{a|x}(\lambda)=\tpa(a|x)q(b|[a|x],y)\,,
\end{aligned}
\end{equation}
and for $a\in[A_x]\backslash\hat{\cA}_x$, we have $p(a,b|x,y)=0$.

\section{Non-signalling and contextual case}\label{sec:nstc}

Given a non-signalling correlation $p$ in a Bell scenario $(\bA,\bB,\cX,\cY)$ we find that $q$ defined by Eq.\vicky{~\eqref{eq:btc}} is in the set of contextual behaviours in the contextuality scenario $\left(\norm{\cA^p},\cY,\bB,\cNS (p_\tA)\right)$ since
\begin{equation}
\sum_{a\in\cA^p_x}p_\tA(a|x)q(b|[a|x],y)=\sum_{a\in\cA^p_x}p(a,b|x,y)=\sum_{a\in[\vicky{A}_x]}p(a,b|x,y)=p_\tB(b|y)\,,
\end{equation}
for all $b\in[B_y]$\vicky{, $x\in[X]$} and $y\in[\cY]$.

Conversely, given a behaviour in the contextual set of a scenario $\left(Z ,\cY,\bB,\cNS (\tpa)\right)$ we find that the correlation $p$ in Eq.~\vicky{~\eqref{eq:ctb}} is non-signalling in the Bell scenario $(\bA,\bB,\cX,\cY)$ since
\begin{equation}
\sum_{a\in[A_x]}p(a,b|x,y)=\sum_{a\in\hat{\cA}_x}q(b|[a|x],y)\tpa(a|x)=Q(b|y)\,,
\end{equation}
for some $Q(b|y)\in\mathbb{R}$ \vicky{(recalling that $\hat{\cA}_x$ denotes the subset of $[A_x]$ such that $\tpa(a|x)>0$)}, and
\begin{equation}
\sum_{b\in[B_y]}p(a,b|x,y)=\sum_{b\in[B_y]}q(b|[a|x],y)\tpa(a|x)=\tpa(a|x)\,,
\end{equation}
for all for all $b\in[B_y]$ and $y\in[\cY]$.

\section{Limitations of the map}\label{sec:counterapp}
Preparation equivalences in the form of Eq.~(14) of the main text generally may involve a single preparation appearing in multiple mixtures, for example, see how preparation $P_1$ appears in all three mixtures in Eq.~\eqref{eq:notbell} below. Such equivalences do not arise from the no-signalling constraint in a remote-preparation scenario. In this section we will demonstrate with an explicit example how treating the multiple instances of a single preparation as different preparations in order to apply our map can result in local correlations being mapped to contextual behaviours. 

Preparation equivalences do not have one unique expression. For example, the equivalence
\begin{equation}\label{eq:notbell}
\frac12 (P_1+P_2)\simeq\frac13 (P_1+P_3+P_4)\,,
\end{equation}
can be expressed as 
\begin{equation}\label{eq:yesbell}
\frac14 P_1+ \frac34 P_2\simeq\frac12 (P_3+P_4)\,.
\end{equation}
Relabelling these preparations can now yield an equivalence in the form \vicky{in Eq.~\eqref{eq:nsequiv}}, since each preparation only appears once.
In fact, any single equivalence $\sum_{a=1}^{Z }p_{a,1}P_a\simeq\sum_{a=1}^{Z }p_{a,2}P_a$ can be expressed such that each preparation only features once, like in the case of Eqs.~\eqref{eq:notbell} and~\eqref{eq:yesbell}. This rearrangement is achieved by subtracting $p_{a,x}P_a$ from both sides and renormalising for all $a$, where $x\in\{1,2\}$ is such that $p_{a,x}=\min\{p_{a,1},p_{a,2}\}$.

Performing this procedure for each of the individual equivalences in Eq.~(14) \vicky{of the main text} will in general lead to equivalences given by decompositions of multiple different hypothetical preparations. For example,
\begin{equation}\label{eq:counter}
\mathcal{OH}=\Bigg\{\frac12 (P_1+P_2)\simeq\frac13 (P_1+P_3+P_4) \simeq\frac15 P_1+ \frac25 (P_3+P_5)\Bigg\}
\end{equation}
becomes
\begin{equation}
\begin{aligned}
\frac14 P_1+ \frac34 P_2&\simeq\frac12 (P_3+P_4)\,\text{ and }\\
\frac38 P_1+\frac58 P_2&\simeq\frac12 (P_3+P_5)\,.
\end{aligned}
\end{equation}
%


Consider then the contextuality scenario $H=(5,2,2,\mathcal{OH})$, where the preparation equivalences \vicky{$\mathcal{OH}$} are given in Eq.~\eqref{eq:counter}. By considering multiple instances of a repeated preparation, e.g.~$P_3$, as distinct preparations, we can embed the behaviours from the scenario $H$ in a scenario $H'$ which can be mapped to a Bell scenario. To do so, we first subtract $\frac15 P_1$ from each hypothetical preparation and renormalise to arrive at the relations
\begin{equation}
\frac{3}{8} P_1 + \frac58 P_2  \simeq \frac16 P_1 + \frac{5}{12}( P_3 + P_4 ) \simeq \frac12( P_3 + P_5 )
\end{equation}
Next, we can treat the two instances of $P_1$ and the two instances of $P_3$ as different preparations (re-interpreting the second instance of $P_1$ as $P_6$ and the second instance of $P_3$ as $P_7$) to embed the behaviours into $H'=(7,2,2,\mathcal{OH}')$ where
\begin{equation}
\mathcal{OH}'=\Bigg\{\frac{3}{8} P_1 + \frac58 P_2  \simeq \frac16 P_6 + \frac{5}{12}( P_3 + P_4 ) \simeq \frac12( P_7 + P_5 )\Bigg\}\,.
\end{equation}
Under relabelling, these preparation equivalences are of the form $\cNS (p_\tA)$. Explicitly, we can map a behaviour $q$ in the scenario $H$ to a behaviour $q'$ in the scenario $H'$ by setting $q'(b|x,y)=q(b|x,y)$ for $1\leq x\leq 5$, and $q'(b|6,y)=q(b|1,y)$ and $q'(b|7,y) = q(b|3,y)$. However, although under this embedding non-contextual and quantum behaviours remain non-contexual and quantum, respectively, it is possible for a contextual behaviour to become non-contextual due to the relaxation of preparation equivalences.

Indeed, we now give an explicit example of a contextual behaviour in $H$ that becomes a non-contextual behaviour in $H'$ using the above notation and mapping. Using the procedure in Ref.~\cite{PhysRevA.97.062103} and the vertex enumeration software package \texttt{lrs}~\cite{lrs}, we found all the facet inequalities defining the non-contextual polytope in $H$. The polytope has 60 facets, one of which is given by the inequality
\begin{equation}\label{eq:ineq_5_prep}
-q(1|1,2) -3 q(1|2,1) + 2 q(1|3,1) + 2 q(1|3,2) + 2 \ge 0.
\end{equation}
An explicit contextual behaviour violating this inequality is given by
\begin{equation}\label{eq:p_5_prep}
\begin{split}
q_c & \left. = \big( q_c(1|1,1), q_c(1|1,2), q_c(1|2,1), q_c(1|,2,2), \ldots, q_c(1|5,2) \big) \right. \\
& \left. = \left( \frac{19}{200},\frac{1}{2},\frac{127}{200},\frac{1}{2},\frac{19}{200},\frac{19}{200},\frac{181}{200},\frac{181}{200},\frac{77}{100},\frac{181}{200},\frac{19}{200},\frac{1}{2},\frac{19}{200},\frac{19}{200} \right)\,, \right.
\end{split}
\end{equation}
where $q_c(2|x,y)=1-q_c(1|x,y)$. In particular, this behaviour violates Eq.~\eqref{eq:ineq_5_prep} by $-\frac{1}{40}$, and thus it is contextual.

Let us now map $q_c$ above to a behaviour $q'_c$ in the scenario $H'$, that is, we have
\begin{equation}\label{eq:p_5_prep_big}
q'_c = \left( \frac{19}{200},\frac{1}{2},\frac{127}{200},\frac{1}{2},\frac{19}{200},\frac{19}{200},\frac{181}{200},\frac{181}{200},\frac{77}{100},\frac{181}{200},\frac{19}{200},\frac{1}{2},\frac{19}{200},\frac{19}{200}, \frac{19}{200},\frac{1}{2}, \frac{19}{200},\frac{19}{200} \right).
\end{equation}
This behaviour is non-contextual, which can be shown by constructing an explicit non-contextual model, that is, a measurable space $(\Lambda, \Sigma)$, probability measures $\mu_x: \Sigma \to [0,1]$ and $\xi_y(b|.)$ response functions such that
\begin{equation}\label{eq:q_c_model}
q'_c(b|x,y) = \int_\Lambda \xi_y(b|\lambda) \mathrm{d}\mu_x(\lambda) \quad \forall b,x,y.
\end{equation}
Let us take the discrete measurable space $\Lambda = \{1,2,3,4\}$ with the usual $\sigma$-algebra $\Sigma$ of the power sets. We define the response functions
\begin{gather}
\xi_1(1|1) = \xi_1(1|3) = 0, \quad \xi_1(1|2) = \xi_1(1|4) = 1 \\
\xi_2(1|1) = \xi_2(1|2) = 0, \quad \xi_2(1|3) = \xi_2(1|4) = 1
\end{gather}
with $\xi_y(2|\lambda) = 1 - \xi_y(1|\lambda)$ for all $y$ and $\lambda$. Furthermore, we define the probability measures via their values $\mu_x(\{\lambda\})$ given by
\begin{equation}
\mu =
\left(
\def\arraystretch{1.5}
\begin{array}{ccccccc}
 \frac{1}{2} & \frac{73}{200} & \frac{1671}{2000} & 0 & 0 & \frac{81}{200} & \frac{133}{160} \\
 \frac{81}{200} & 0 & \frac{139}{2000} & \frac{19}{200} & \frac{23}{100} & \frac{1}{2} & \frac{59}{800} \\
 0 & \frac{27}{200} & \frac{139}{2000} & \frac{19}{200} & \frac{19}{200} & \frac{19}{200} & \frac{59}{800} \\
 \frac{19}{200} & \frac{1}{2} & \frac{51}{2000} & \frac{81}{100} & \frac{27}{40} & 0 & \frac{17}{800} \\
\end{array}
\right),
\end{equation}
where the rows are indexed by $\lambda$ and the columns are indexed by $x$, i.e.~$\mu_{\lambda,x}=\mu_x(\{\lambda\})$. It is a straightforward computation to verify Eq.~\eqref{eq:q_c_model} with these choices. Thus, the contextual behaviour $q_c$ in the scenario $H$ is mapped to a non-contextual behaviour $q'_c$ in the scenario $H'$.

\section{Proof of Corollary 1}\label{sec:proof_membership}

Suppose there exists an algorithm to decide whether any behaviour belongs to the quantum set in any given contextuality scenario. Then, given any correlation $p$ in a Bell scenario $(\bA,\bB,\cX,\cY)$ one can decide whether $q$ given by Eq.\vicky{~\eqref{eq:btc}} belongs to the quantum set in the contextuality scenario $(\norm{\cA^p},\cY,\bB,\cNS (p_\tA))$. Since $q\in\cQ$ if and only if $p\in\cC_{qs}$, one could therefore decide the membership problem for the set of quantum spatial correlations, however this problem is known to be undecidable~\cite{ji2020mip,slofstra_2019}. 

\section{Proof of Corollary 2}\label{sec:proof_finite-dim}

 If any behaviour in $\cQ$ in any contextuality scenario could be realised with finite dimensional quantum systems, then the construction in Sec.~\ref{sec:ctb} would give a finite dimensional quantum realisation of any correlation in a Bell scenario, which is known not to exist~\cite{coladangelo2020inherently}.

\section{Proof of Corollary 3}\label{sec:proof_nclosed}

For the proofs of Corollaries~3 and~4, we will find it useful to remove probability zero or one outcomes of Alice from a correlation. To do so we will map a given correlation in a Bell scenario to one in a scenario with fewer inputs and/or outputs in which Alice's marginal probabilities are strictly between zero and one. We now describe this map and show how it preserves the closure $\cC_{qa}$ of the set of quantum spatial correlations $\cC_{qs}$.  

Let $\hat{p}\in\cC_{qa}$ be a correlation in a Bell scenario $(\bA,\bB,\cX,\cY)$ such that $\hat{p}_\tA(a|x)=0$ for all $a> a'_x$ for each $x\in[\cX]$ (note that if there are some zeroes in Alice's marginals, we can always relabel Alice's outcomes such that these zeroes appear at $a > a'_x$, since relabelling is a symmetry of $\cC_{qa}$). Furthermore, let Alice's outcomes be completely deterministic for inputs $x>X'$. Since $\hat{p}\in\cC_{qa}$, there exists a sequence of correlations $(\hat{p}_j)_j\subset\cC_{qa}$ with finite dimensional quantum realisations such that $\hat{p}_j\to \hat{p}$ as $j\to\infty$, i.e, for each $j\in\mathbb{N}$ we have 
\begin{equation}\label{eq:seq}
\hat{p}_j(a,b|x,y)=\bra{\psi_j}M^{x,j}_a\otimes N^{y,j}_b\ket{\psi_j}
\end{equation}
for some separable Hilbert spaces $\cH^j_\tA$ and $\cH^j_\tB$, unit vectors $\ket{\psi_j} \in\cH_\tA^j\otimes\cH_\tB^j$ and projective measurements $\{M^{x,j}_a\}$ and $\{N^{y,j}_b\}$ on $\cH^j_\tA$ and $\cH^j_\tB$, respectively. 

Now, consider the Bell scenario $(\bA',\bB,\cX',\cY)$ in which we have removed all of Alice's inputs that give a deterministic outcome in $\hat{p}$, i.e.~$x>X'$ and all the outputs $a|x$ for Alice such that $\hat{p}_\tA(a|x)=0$, i.e.~$a > a'_x$. Therefore, $\bA'$ has elements $A'_{x}=a'_x$. Define a correlation $\tau(\hat{p})=p'$ in this scenario by $p'(a,b|x,y)=\hat{p}(a,b|x,y)$ for all $a\in[A'_x]$, $b\in[B_y]$, $x\in[\cX']$ and $y\in[\cY]$. 

\begin{lem}\label{lem:Cqa_small}
A correlation $\hat{p}$ is in the set $\cC_{qa}$ of a Bell scenario $(\bA,\bB,\cX,\cY)$ if and only if $p'=\tau(\hat{p})$ is a correlation in the set $\cC_{qa}$ of the Bell scenario $(\bA',\bB,\cX',\cY)$
\end{lem} 

\begin{proof}
Firstly, it is clear that removing all the deterministic inputs $x>X'$ from each correlation in the sequence Eq.~\eqref{eq:seq} leaves a sequence of correlations $p_j(a,b|x,y)=\hat{p}_j(a,b|x,y)$ for all $a\in[A_x]$, $b\in[B_y]$, $x\in[\cX']$ and $y\in[\cY]$ with a quantum realisation that tend to a correlation $p$ defined by $p(a,b|x,y)=\hat{p}(a,b|x,y)$ for all $a\in[A_x]$, $b\in[B_y]$, $x\in[\cX']$ and $y\in[\cY]$.

We then proceed by removing all the zero probability outcomes of one input $x^*\in[X']$ of Alice, by mapping $p$ to a correlation $p^*$ in the scenario $(\bA^*,\bB,X',Y)$, where $A^*_{x^*}=a'_{x^*}$ and $A^*_x=A_x$ for all $x\neq x^*$. We take $p^*(a,b|x,y)=p(a,b|x,y)$ for all $a\in[A^*_x]$, $b\in[B_y]$, $x\in[X']$ and $y\in[Y]$. Let $\Pi^j=\sum_{a\leq A^*_{x^\ast}}M^{x,j}_a$ and $\cH_\tA^{j*}$ be the support of $\Pi^j$. Observe that, denoting the identity operator on $\cH_\tA^j$ ($\cH_\tB^j$) by $\I_{\tA^j}$ ($\I_{\tB^j}$), we have that
\begin{equation}
0=\sum_{a>A^*_{x^*}}p_\tA(a|x^*)=\lim_{j\to\infty}\bra{\psi_j}(\I_{\tA^j}-\Pi^j)\otimes\I_{\tB^j}\ket{\psi_j}=\lim_{j\to\infty} \left( 1-\bra{\psi_j}\Pi^j\otimes\I_{\tB^j}\ket{\psi_j} \right) ,
\end{equation}
which gives 
\begin{equation}\label{eq:lim1}
\lim_{j\to\infty}\bra{\psi_j}\Pi^j\otimes\I_{\tB^j}\ket{\psi_j}=1.
\end{equation}

We may then define the states 
\begin{equation}
\ket{\psi_j^*}=\frac{\Pi^j\otimes\I_{\tB^j}\ket{\psi_j}}{\sqrt{\bra{\psi_j}\Pi^j\otimes\I_{\tB^j}\ket{\psi_j}}}\in\cH_\tA^{j*}\otimes\cH^j_\tB\,,
\end{equation}  
where without loss of generality we can assume that the denominator is strictly positive for all $j\in\mathbb{N}$ since it tends to one in the limit $j\to\infty$. 

It also follows from Eq.~\eqref{eq:lim1} that (embedding $\ket{ \psi^*_j }$ into $\cH_\tA^j \otimes \cH_\tB^j$)
\begin{equation}\label{eq:prodlim}
\lim_{j\to\infty}\braket{\psi_j|\psi^*_j}=\lim_{j\to\infty}\frac{\bra{\psi_j}\Pi^j\otimes\I_{\tB^j}\ket{\psi_j}}{\sqrt{\bra{\psi_j}\Pi^j\otimes\I_{\tB^j}\ket{\psi_j}}}=1\,. 
\end{equation}
Observe that we can find unit vectors $\ket{ \psi^\perp_j } \in \cH_\tA^{j}\otimes\cH^j_\tB$ orthogonal to $\ket{ \psi_j }$ such that (again for the embedding)
\begin{equation}\label{eq:decomp}
\ket{\psi^*_j}=\alpha_j\ket{\psi_j}+\beta_j \ket{\psi^\perp_j}
\end{equation}
for each $j\in\mathbb{N}$. Note that we can choose the $\ket{\psi^\perp_j}$ such that $\beta_j \in \mathbb{R}$, and we have that $\alpha_j=\bra{\psi_j}\Pi^j\otimes\I_{\tB^j}\ket{\psi_j}\in\mathbb{R}$, which also implies $\beta_j = \sqrt{1-\alpha_j^2}$. By Eq.~\eqref{eq:prodlim}, we find that $\lim_{j\to\infty}\alpha_j=1$ and hence, also $\lim_{j\to\infty}\sqrt{1-\alpha_j^2}=0$.

Now, consider the sequence of quantum spatial correlations in the scenario $(\bA^*,\bB,\cX',\cY)$ given by 
\begin{equation}
p^*_j(a,b|x,y)=\bra{\psi^*_j}R^{x,j}_a\otimes N^{y,j}_b\ket{\psi^*_j}\,,
\end{equation}
for all $a\in[A^*_x]$, $b\in[B_y]$, $x\in[X']$ and $y\in[Y]$, where the operators $R^{x,j}_a=\Pi^jM^{x,j}_a\Pi^j$ form projective measurements on the subspace $\cH_\tA^{j*}$ of $\cH_\tA^{j}$ since they satisfy $\sum_{a\in [A^*_x]}R^{x,j}_a=\sum_{a\in [A^*_x]}\Pi^jM^{x,j}_a\Pi^j=\Pi^j\I_{\tA^j}\Pi^j=\I_{\tA^{*j}}$.

We can now evaluate the limit of our sequence of correlations using the expression in Eq.~\eqref{eq:decomp}:
\begin{equation}\label{eq:seqlim}
\begin{aligned}
\lim_{j\to\infty}p^*_j(a,b|x,y)=&\lim_{j\to\infty}\bra{\psi^*_j}R^{x,j}_a\otimes N^{y,j}_b\ket{\psi^*_j}\\
=&\lim_{j\to\infty}\bra{\psi^*_j}M^{x,j}_a\otimes N^{y,j}_b\ket{\psi^*_j}\\
=&\lim_{j\to\infty}\alpha_j^2\bra{\psi_j}M^{x,j}_a\otimes N^{y,j}_b\ket{\psi_j}\\&+2\mathrm{Re}\left(\alpha_j\sqrt{1-\alpha_j^2}\bra{\psi^\perp_j}M^{x,j}_a\otimes N^{y,j}_b\ket{\psi_j}\right)\\&+(1-\alpha_j^2)^2\bra{\psi^\perp_j}M^{x,j}_a\otimes N^{y,j}_b\ket{\psi^\perp_j}\,.
\end{aligned}
\end{equation}
Since we have $0\leq M^{x,j}_a\otimes N^{y,j}_b\leq\I_{\tA^{j}\tB^j}$, both $\abs{\bra{\psi^\perp_j}M^{x,j}_a\otimes N^{y,j}_b\ket{\psi^\perp_j}}$ and $\abs{\bra{\psi^\perp_j}M^{x,j}_a\otimes N^{y,j}_b\ket{\psi_j}}$ are bounded in the unit interval. It follows that the final two summands of the last expression in Eq.~\eqref{eq:seqlim} tend to zero in the limit due to the factor of $\sqrt{1-\alpha_j^2}$. Remembering that $\alpha_j\to1$ as $j\to\infty$, we are left with 
\begin{equation}
\lim_{j\to\infty}p^*_j(a,b|x,y)=\lim_{j\to\infty}\bra{\psi_j}M^{x,j}_a\otimes N^{y,j}_b\ket{\psi_j}=
p(a,b|x,y)\,,
\end{equation}
for all $a\in[A^*_x]$, $b\in[B_y]$, $x\in[X']$ and $y\in[Y]$. This argument can be applied iteratively for each $x\in[X']$ to show that the correlation $p'=\tau(\hat{p})$ is a member of the set $\cC_{qa}$ in the Bell scenario $(\bA',\bB,\cX',\cY)$.

Conversely, given any correlation $p'\in\cC_{qa}$ in a scenario $(\bA',\bB,\cX',\cY)$ we can embed the correlation in a scenario $(\bA,\bB,\cX,\cY)$ in which Alice has more inputs and/or outputs via the map
\begin{equation}
\hat{p}(a,b|x,y)=\begin{cases}p'(a,b|x,y) &\text{for }a\in[A'_x],\, b\in[B_y],\, x\in[\cX'],\, y\in[\cY]\\
1&\text{for }a=1\text{ and }X'<x\leq X\\
0&\text{otherwise.}
\end{cases}
\end{equation}
If $p'\in\cC_{qa}$ then there exists a sequence of correlations $p'_j$ which tend to $p'$ and have quantum realisations. This sequence can be transformed to a sequence of quantum spatial correlations in the scenario $(\bA,\bB,\cX,\cY)$ which tends to $\hat{p}$ by adding zero operators to the POVMs for the additional probability zero outcomes in the existing inputs of Alice and adding POVMs given by the identity operator followed by zero operators for the additional deterministic settings.
\end{proof} 

\begin{rem}\label{rem:Cqs_small}
One can analogously show that a correlation $\hat{p}$ is in the set $\cC_{qs}$ of a Bell scenario $(\bA,\bB,\cX,\cY)$ if and only if $p'=\tau(\hat{p})$ is a correlation in the set $\cC_{qs}$ of the Bell scenario $(\bA',\bB,\cX',\cY)$ with an argument that follows the proof of Lemma~\ref{lem:Cqa_small} but with the simplification of not having to consider limits of sequences of correlations.
\end{rem}

Now, let $p$ be a correlation in a Bell scenario $(\bA,\bB,\cX,\cY)$ that is contained in the closure $\cC_{qa}$ of the set $\cC_{qs}$ of quantum spatial correlations but that is not contained in the set $\cC_{qs}$ itself, i.e. $p\in\cC_{qa}\backslash\cC_{qs}$~\cite{slofstra_2019}. It follows from Lemma \ref{lem:Cqa_small} and Remark \ref{rem:Cqs_small} that $p'=\tau(p)\in\cC_{qa}\backslash\cC_{qs}$ in the scenario $(\bA',\bB,\cX',\cY)$. First, we will construct a sequence $( p^n )_{n \in \mathbb{N}}$ of correlations in $\cC_{qs}$ in the scenario $(\bA',\bB,\cX',\cY)$ converging to $\tp$ such that every element of the sequence has the same marginals for Alice as $\tp$, i.e.~$p^n_\tA = \tp_\tA$ for all $n \in \mathbb{N}$. Since the correlations $p^n$ will all have the same marginals for Alice, they will each be mapped to a behaviour $q^n=\vicky{\Gamma}(p^n)$ in the same single contextuality scenario $(\norm{\cA'^p},\cY,\bB,\cNS(\tp_\tA))$ where $\norm{\cA'^p}=\sum_{x\in[X']}\abs{\cA^p_x}=\sum_{x\in[X']}\abs{\cA^{p'}_x}$. 

Next, we will show that this sequence of behaviours converges to $q=\vicky{\Gamma}(p')$, meaning $q$ is in the closure $\overline{\cQ}$ of the set of quantum behaviours. Finally, it follows that $q\notin\cQ$ since otherwise we could construct a quantum realisation of the correlation $p'=\vicky{\Gamma}^{-1}(q)$, via the method in Sec.~\ref{sec:ctb}. Thus, we have that $q\in\overline{\cQ}\setminus\cQ$.

Consider the correlation $p^{\mathrm{int}}(a,b|x,y) = \tp_\tA(a|x) \frac{1}{B_y}$. We demonstrate that this correlation is in the relative interior of the local polytope (and, hence, in the relative interior of $\cC_{qs}$) as follows. The vertices of the local polytope are exactly the correlations $V$ that admit an expression
\begin{equation}
V(a,b|x,y)=v^\tA(a|x)v^\tB(b|y)\,,
\end{equation} 
for two deterministic conditional probability distributions $v^\tA$ and $v^\tB$. In any polytope of unconstrained conditional probability distributions $r(c|z)$ over some variables $c\in[C]$ and $z\in[Z ]$ for some $C,Z \in\mathbb{N}$, any point on the boundary contains at least one zero element, i.e.~$r(c|z)=0$ for some $c\in[C]$ and $z\in[Z ]$. Both distributions $p^{\mathrm{int}}_\tA(a|x)=\tp_\tA(a|x)$ and $p^{\mathrm{int}}_\tB(b|y)=\frac{1}{B_y}$ are entirely non-zero and, therefore, are in the relative interiors of their respective polytopes. 

It follows that $p^{\mathrm{int}}_\tA(a|x)$ and $p^{\mathrm{int}}_\tB(b|y)$ admit convex decompositions $p^{\mathrm{int}}_\tA(a|x)=\sum_{a,x}\alpha_{a,x}v^\tA(a|x)$ and $p^{\mathrm{int}}_\tB(b|y)=\sum_{b,y}\beta_{b,y}v^\tB(b|y)$ where $\alpha_{a,x}>0$, $\beta_{b,y}>0$ for all $a\in[A'_x]$, $b\in[B_y]$, $x\in[\cX]$ and $y\in[\cY]$ and $\sum_{a,x}\alpha_{a,x}=\sum_{b,y}\beta_{b,y}=1$. Thus, we find that $p^{\mathrm{int}}(a,b|x,y)$ also admits a convex decomposition in which all of the vertices of the local polytope have a strictly non-zero coefficient, namely,
\begin{equation}
p^{\mathrm{int}}(a,b|x,y)=\sum_{a,b,x,y}\alpha_{a,x}\beta_{b,y}v^\tA(a|x)v^\tB(b|y)\,,
\end{equation}
showing that $p^{\mathrm{int}}$ is in the relative interior of the local polytope.

Additionally, we have that $p^{\mathrm{int}}_\tA = \tp_\tA$. Define $p^n = \frac1n p^{\mathrm{int}} + \left( 1 - \frac1n \right) \tp$. Each term of the sequence has the same marginals for Alice, $p^n_\tA=\tp_\tA$, and is in the relative interior of $\cC_{qa}$, since it is a mixture of a point in $\cC_{qa}$ and a point in the relative interior of $\cC_{qa}$. It follows that $(p^n)_{n\in\mathbb{N}}$ is a sequence of points in $\cC_{qs}$ that converge to $\tp$. 

Now, consider the image $(q^n)_{n\in\mathbb{N}}$ of the sequence $(p^n)_{n\in\mathbb{N}}$ under our map $\vicky{\Gamma}$ [see Eq.\vicky{~\eqref{eq:btc}}] in the contextuality scenario $(\norm{\cA'^p},\cY,\bB,\cNS(\tp_\tA))$, noting that all points in the sequence are mapped to the same contextuality scenario since they have the same marginals $\tp_\tA(a|x)$ for Alice. Since each $p^n\in\cC_{qs}$, we have that each $q^n\in\cQ$. Furthermore, we have that for every $\epsilon>0$ there exists $N_\epsilon\in\mathbb{N}$ such that $\norm{p^n-\tp}_1<\epsilon$ for all $n\ge N_\epsilon$. Thus, letting $\epsilon'=\min_{a,x}\{\tp_\tA(a|x)\}\epsilon$, we have that for all $n\ge N_{\epsilon'}$

\begin{equation}
\begin{aligned}
\norm{q^n-q}_1&=\sum_{a,b,x,y}\abs{q^n(b|[a|x],y)-q(b|[a|x],y)}\\
&=\sum_{a,b,x,y}\abs{\frac{p^n(a,b|x,y)}{\tp_\tA(a|x)}-\frac{\tp(a,b|x,y)}{\tp_\tA(a|x)}}\\
& \le \max_{a,x}\left\{\frac{1}{\tp_\tA(a|x)}\right\} \sum_{a,b,x,y} \abs{ p^n(a,b|x,y) - \tp(a,b|x,y) } \\
&< \max_{a,x}\left\{\frac{1}{\tp_\tA(a|x)}\right\}\epsilon' = \frac{\epsilon'}{ \min_{a,x} \{ \tp_\tA(a|x) \} } =\epsilon\,,
\end{aligned}
\end{equation}
and we find that $(q^n)\subset\cQ$ converges to $q$, meaning $q\in\overline{\cQ}$.

On the other hand, we have that $q\notin\cQ$ since otherwise we could construct a quantum realisation of the correlation $p'=\vicky{\Gamma}^{-1}(q)$, via the method in Sec.~\ref{sec:ctb}. Therefore, it follows that that $q\in\overline{\cQ}\setminus\cQ$ and $\cQ$ is not closed.

\section{Proof of Corollary 4}\label{sec:proof_compute}

We require two results from the literature. Firstly, it is known that the following weak-membership problem for $\cC_{qa}=\overline{\cC_{qs}}$ is undecidable~\cite{ji2020mip}: 
\begin{itemize}
\item[[WMEM\!\!]] given a correlation $p$ in a Bell scenario $(\bA,\bB,\cX,\cY)$ and $\varepsilon>0$ decide whether $p\in \cC_{qa}$ or $p\notin \cC^\varepsilon_{qa}=\{p : \,\norm{p-p'}_1 \leq \varepsilon\text{ for some }p'\in \cC_{qa}\}$ with the promise that either $p\in \cC_{qa}$ or $p\notin \cC^\varepsilon_{qa}$,
\end{itemize}
where $\norm{.}_1$ is the $\ell_1$-norm. 

Secondly, it is known~\cite{ji2020mip} that for any $\delta>0$ there exists an algorithm (FIN) that verifies that a correlation $p\in \cC^\delta_{qa}$ and halts for any correlation $p\in \cC^\delta_{qa}$. The input of this algorithm is a fixed correlation $p$. At the $d$-th step, (FIN:$d$), of the algorithm a finite set of quantum correlations, $\{ \tilde{p}^d_n \}_n$, is constructed such that these correlations are realisable when $\mathrm{dim}(\cH_\tA)=\mathrm{dim}(\cH_\tB)=d$, and moreover, for any correlation, $\tilde{p}$ that is also realisable with such Hilbert spaces we have $\norm{\tilde{p}-\tilde{p}^d_n}_1<\delta$ for some $\tilde{p}^d_n$. This can be achieved, because the set of quantum correlations achievable in a fixed dimension is compact. Then, the distance $\norm{p-\tilde{p}^d_n}_1$ is calculated for all $n$. If this distance is less than $\delta$ for some $n$ the algorithm returns $p\in \cC^\delta_{qa}$. Otherwise, the algorithm proceeds to (FIN:$d+1$). Since the closure of the set of quantum correlations realisable in some finite dimension is the same as $\cC_{qa}$~\cite{scholz2008tsirelson, Fri12} , it follows that if $p\in \cC^\delta_{qa}$ then the algorithm halts for some finite $d$, otherwise it does not halt.

Now, suppose there exists a hierarchy of SDPs wherein each level, $j$, decides whether a behaviour $q$ in a contextuality scenario $(Z ,\cY,\bB,\cNS (\tpa))$ is in a superset $\cQ_j$ of $\cQ$ or not, and these supersets converge to $\cQ$ as $j$ tends to infinity (i.e. $\cQ=\cap_{\mathbb{N}}\cQ_j$). Under this hypothesis, we will construct an algorithm that decides the problem~[WMEM], and therefore reach a contradiction.

We can now give the algorithm that decides [WMEM]. Given a correlation $p$ with the promise that either $p\in \cC_{qa}$ or $p\notin \cC^\varepsilon_{qa}$:

\begin{itemize}
\item[Step (1):] If $p_\tA$ is deterministic return $p\in\cC_{qa}$ and halt. \\
Otherwise, relabel Alice's inputs such that any inputs giving a deterministic outcome are labelled with the highest values $x>X'$ in $[X]$ and for each $x\in[X']$ relabel the outcomes such that any zeroes in $p_A$ are for outcomes $a> a'_x$ for some $a'_x\in\mathbb{N}$ and  (we retain the notation $p$ for this relabelling), and map $p$ to the correlation $\tp=\tau(p)$ in the Bell scenario $(\bA',\bB,\cX',\cY)$ such that $\tp$ has marginals $\tp_\tA$ strictly between zero and one (see Sec.~\ref{sec:proof_nclosed}).
\item[Step (2):] Map $\tp$ to $q=\vicky{\Gamma}(\tp)$ in the contextuality scenario $(\norm{\cA'^p},\cY,\bB,\cNS(\tp_\tA))$, where $\norm{\cA'^p}=\sum_{x\in[X']}\abs{\cA^p_x}=\sum_{x\in[X']}\abs{\cA^{p'}_x}$---see \vicky{Eqs.~\eqref{eq:Gam} and~\eqref{eq:btc}.}.
\item[Step ($j\geq3$):] Use level $j$ of the SDP hierarchy to decide whether $q\in\cQ_j$. \\
If $q\notin\cQ_j$: return $p\notin C^\varepsilon_{qa}$ and halt. \\
If $q\in\cQ_j$: run step (FIN:$j$) of the algorithm (FIN) on $\tp$ with $\delta<\varepsilon$. \\
If (FIN:$j$) returns $\tp\in \cC^\delta_{qa}$: return $p\in \cC_{qa}$ and halt. \\
Otherwise, perform Step ($j$+1).
\end{itemize}

To see that the algorithm would return the correct answer, first, define $\cQ^{\varepsilon}=\{q :\norm{q-q'}_1 \leq \varepsilon\text{ for some }q'\in \cQ\}$.

{\bf Case (1) $p\notin\cC^\varepsilon_{qa}$}: \\
Let $\tp_\varepsilon$ be any behaviour such that $\norm{\tp-\tp_\varepsilon}_1<\varepsilon$. Then we have that $\norm{p-\tau^{-1}(\tp_\varepsilon)}_1=\norm{\tp-\tp_\varepsilon}_1<\varepsilon$ and thus $\tau^{-1}(\tp_\varepsilon)\notin\cC_{qa}$. Therefore, $\tp_\varepsilon\notin\cC_{qa}$ in the scenario $(\bA',\bB,\cX',\cY)$ and $\tp\notin\cC^\varepsilon_{qa}$, since there is an $\varepsilon$-ball around $p'$ entirely outside of $\cC_{qa}$.

Now we will show that $q=\vicky{\Gamma}(\tp)\notin\cQ^{\varepsilon}$. Consider any behaviour $q_\varepsilon$ such that $\norm{q_\varepsilon-q}_1<\varepsilon$ and let $p_\varepsilon = \vicky{\Gamma}^{-1}(q_\varepsilon)$ be the image of $q_\varepsilon$ under the map in Eq.\vicky{~\eqref{eq:ctb}} (where $\tpa=\tp_\tA$). Then we have
\begin{equation}
\begin{aligned}
\norm{p_\varepsilon-\tp}_1&=\sum_{a,b,x,y}\abs{\tp_\tA(a|x)q_\varepsilon(b|[a|x],y)-\tp(a,b|x,y)}\\
&=\sum_{a,b,x,y}\abs{\tp_\tA(a|x)q_\varepsilon(b|[a|x],y)-\tp_\tA(a|x)q(b|[a|x],y)}\\
&=\sum_{a,b,x,y}\tp_\tA(a|x)\abs{q_\varepsilon(b|[a|x],y)-q(b|[a|x],y)} \\
&\le \sum_{a,b,x,y} \abs{q_\varepsilon(b|[a|x],y)-q(b|[a|x],y)} <\varepsilon\,.
\end{aligned}
\end{equation}  
Therefore, we have that $p_\varepsilon\notin\cC_{qa}$, which implies $p_\varepsilon\notin\cC_{qs}$ and thus $q_\varepsilon = \vicky{\Gamma}(p_\varepsilon) \notin\cQ$. We have shown that there is an $\varepsilon$-ball around $q$ entirely outside of $\cQ$, thus $q\notin\cQ^{\varepsilon}$. 

For a finite level $j$ of the SDP hierarchy we find $q\notin\cQ_j$ and hence, at a finite step $(j+1)$ of the algorithm we obtain $p\notin\cC^\varepsilon_{qa}$. On the other hand, at no Step $(j)$ will the algorithm return $p\in\cC^{qa}$, since we have shown that $\tp\notin\cC^\varepsilon_{qa}$ in the scenario $(\bA',\bB,\cX',\cY)$, and therefore $\norm{\tp-p''}_1>\varepsilon>\delta$ for any correlation $p''$ realisable with finite dimensional quantum systems.

{\bf Case (2) $p\in\cC_{qa}$:} \\
We have that $\tp=\tau(p)\in \cC_{qa}$ in the scenario $(\bA',\bB,\cX',\cY)$, and we will show that $q = \vicky{\Gamma}(p') \in\overline{\cQ}$ using the same method as in Sec.~\ref{sec:proof_nclosed}. To do so, we will construct a sequence $( p^n )_{n \in \mathbb{N}}$ of correlations in $\cC_{qs}$ converging to $\tp$ such that every element of the sequence has the same marginals for Alice as $\tp$, i.e.~$p^n_\tA = \tp_\tA$ for all $n \in \mathbb{N}$. Since the correlations $p^n$ will all have the same marginals for Alice, they will all be mapped to a behaviour $q^n=\Lambda(p^n)$ in the same contextuality scenario $(\norm{\cA'^p},\cY,\bB,\cNS(\tp_\tA))$. 

Consider the correlation $p^{\mathrm{int}}(a,b|x,y) = \tp_\tA(a|x) \frac{1}{B_y}$. We demonstrate that this correlation is in the relative interior of the local polytope (and, hence, in the relative interior of $\cC_{qs}$) as follows. The vertices of the local polytope are exactly the correlations $V$ that admit an expression
\begin{equation}
V(a,b|x,y)=v^\tA(a|x)v^\tB(b|y)\,,
\end{equation} 
for two deterministic, conditional probability distributions $v^\tA$ and $v^\tB$. In any polytope of unconstrained conditional probability distributions $r(c|z)$ over some variables $c\in[C]$ and $z\in[Z ]$ for some $C,Z \in\mathbb{N}$, any point on the boundary contains at least one zero element, i.e.~$r(c|z)=0$ for some $c\in[C]$ and $z\in[Z ]$. Both distributions $p^{\mathrm{int}}_\tA(a|x)=\tp_\tA(a|x)$ and $p^{\mathrm{int}}_\tB(b|y)=\frac{1}{B_y}$ are entirely non-zero and, therefore, are in the relative interiors of their respective polytopes. 

It follows that $p^{\mathrm{int}}_\tA(a|x)$ and $p^{\mathrm{int}}_\tB(b|y)$ admit convex decompositions $p^{\mathrm{int}}_\tA(a|x)=\sum_{a,x}\alpha_{a,x}v^\tA(a|x)$ and $p^{\mathrm{int}}_\tB(b|y)=\sum_{b,y}\beta_{b,y}v^\tB(b|y)$ where $\alpha_{a,x}>0$, $\beta_{b,y}>0$ for all $a\in[A'_x]$, $b\in[B_y]$, $x\in[\cX]$ and $y\in[\cY]$ and $\sum_{a,x}\alpha_{a,x}=\sum_{b,y}\beta_{b,y}=1$. Thus, we find that $p^{\mathrm{int}}(a,b|x,y)$ also admits a convex decomposition in which all of the vertices of the local polytope have a strictly non-zero coefficient, namely,
\begin{equation}
p^{\mathrm{int}}(a,b|x,y)=\sum_{a,b,x,y}\alpha_{a,x}\beta_{b,y}v^\tA(a|x)v^\tB(b|y)\,,
\end{equation}
showing that $p^{\mathrm{int}}$ is in the relative interior of the local polytope.

Additionally, we have that $p^{\mathrm{int}}_\tA = \tp_\tA$. Define $p^n = \frac1n p^{\mathrm{int}} + \left( 1 - \frac1n \right) \tp$. Each term of the sequence has the same marginals for Alice, $p^n_\tA=\tp_\tA$, and is in the relative interior of $\cC_{qa}$, since it is a mixture of a point in $\cC_{qa}$ and a point in the relative interior of $\cC_{qa}$. It follows that $(p^n)_{n\in\mathbb{N}}$ is a sequence of points in $\cC_{qs}$ that converge to $\tp$. 

Now consider the image $(q^n)_{n\in\mathbb{N}}$ of the sequence $(p^n)_{n\in\mathbb{N}}$ under our map $\vicky{\Gamma}$ [see Eq.~\vicky{\eqref{eq:btc}}] in the contextuality scenario $(\norm{\cA'^p},\cY,\bB,\cNS(\tp_\tA))$, noting that all points in the sequence are mapped to the same contextuality scenario since they have the same marginals $\tp_\tA(a|x)$ for Alice. Since each $p^n\in\cC_{qs}$, we have that each $q^n\in\cQ$. Furthermore, we have that for every $\epsilon>0$ there exists $N_\epsilon\in\mathbb{N}$ such that $\norm{p^n-\tp}_1<\epsilon$ for all $n\ge N_\epsilon$. Thus, letting $\epsilon'=\min_{a,x}\{\tp_\tA(a|x)\}\epsilon$, we have that for all $n\ge N_{\epsilon'}$

\begin{equation}
\begin{aligned}
\norm{q^n-q}_1&=\sum_{a,b,x,y}\abs{q^n(b|[a|x],y)-q(b|[a|x],y)}\\
&=\sum_{a,b,x,y}\abs{\frac{p^n(a,b|x,y)}{\tp_\tA(a|x)}-\frac{\tp(a,b|x,y)}{\tp_\tA(a|x)}}\\
& \le \max_{a,x}\left\{\frac{1}{\tp_\tA(a|x)}\right\} \sum_{a,b,x,y} \abs{ p^n(a,b|x,y) - \tp(a,b|x,y) } \\
&< \max_{a,x}\left\{\frac{1}{\tp_\tA(a|x)}\right\}\epsilon' = \frac{\epsilon'}{ \min_{a,x} \{ \tp_\tA(a|x) \} } =\epsilon\,,
\end{aligned}
\end{equation}
and we find that $(q^n)\subset\cQ$ converges to $q$, meaning $q\in\overline{\cQ}$.

Therefore, in this case the algorithm will not return $p\notin \cC^\varepsilon_{qa}$ in any step $j$. On the other hand, at some finite step (FIN:$d$) the algorithm (FIN) will establish $\tp\in \cC^\delta_{qa}$ and our algorithm will return $p\in \cC_{qa}$ and halt at Step $d$.

\end{document}